\documentclass[a4paper,reqno]{amsart}

\makeatletter
\@ifclassloaded{llncs}{}{\usepackage{a4wide}\usepackage[foot]{amsaddr}}
\makeatother
\usepackage{amsmath}
\usepackage{amssymb}
\usepackage{stmaryrd}
\usepackage[english]{babel}
\addto\captionsenglish{}

\usepackage{color}
\usepackage{enumitem}
\usepackage{ellipsis}
\usepackage{hyperref}
\usepackage[misc,geometry]{ifsym}

\usepackage{makecell}
\usepackage[mathscr]{eucal}

\usepackage{mathtools}

\usepackage{etoolbox}

\usepackage{calrsfs}

\usepackage[dvipsnames]{xcolor}

\usepackage[normalem]{ulem}

\usepackage{tikz}
\usetikzlibrary{matrix}

\usepackage{url}

\usepackage{lipsum}

\makeatletter
\def\@tocline#1#2#3#4#5#6#7{\relax
  \ifnum #1>\c@tocdepth \else
    \par \addpenalty\@secpenalty\addvspace{#2}\begingroup \hyphenpenalty\@M
    \@ifempty{#4}{\@tempdima\csname r@tocindent\number#1\endcsname\relax
    }{\@tempdima#4\relax
    }\parindent\z@ \leftskip#3\relax \advance\leftskip\@tempdima\relax
    \rightskip\@pnumwidth plus4em \parfillskip-\@pnumwidth
    #5\leavevmode\hskip-\@tempdima #6\nobreak\relax
    \ifnum#1<0\hfill\else\dotfill\fi\hbox to\@pnumwidth{\@tocpagenum{#7}}\par
    \nobreak
    \endgroup
  \fi}
\makeatother

 \renewenvironment{proof}{{\noindent \itshape
    Proof.}}{\qed\vspace{\baselineskip}}

\newtoggle{llncs}

\makeatletter
\@ifclassloaded{llncs}{\toggletrue{llncs}}{\togglefalse{llncs}}
\makeatother

\iftoggle{llncs}{}{
  \theoremstyle{plain}
  \newtheorem{definition}{Definition}[section]
  \newtheorem{problem}[definition]{Problem}
  \newtheorem{theorem}[definition]{Theorem}
  \newtheorem{example}[definition]{Example}
  \newtheorem{proposition}[definition]{Proposition}
  \newtheorem{lemma}[definition]{Lemma}
  \theoremstyle{remark}
  \newtheorem{remark}[definition]{Remark}
}

\renewcommand{\vec}[1]{\mathbf{#1}}

\newcommand{\code}[1]{\mathscr{#1}}

\newcommand{\word}[1]{\ensuremath{\boldsymbol{\rm #1}}}
\newcommand{\bfa}{\word{a}}
\newcommand{\bfb}{\word{b}}
\newcommand{\bfc}{\word{c}}
\newcommand{\bfe}{\word{e}}
\newcommand{\bfg}{\word{g}}
\newcommand{\bfh}{\word{h}}

\newcommand{\bfr}{\word{r}}
\newcommand{\bfs}{\word{s}}
\newcommand{\bfu}{\word{u}}
\newcommand{\bfv}{\word{v}}

\newcommand{\mat}[1]{\ensuremath{\boldsymbol{\rm #1}}}
\newcommand{\bfA}{\mat{A}}

\newcommand{\bfG}{\mat{G}}

\newcommand{\Enc}{\mathsf{Enc}}

\newcommand{\goth}[1]{\ensuremath{\mathfrak{#1}}}

\newcommand{\gothp}{\goth{p}}
\newcommand{\gothP}{\goth{P}}

\newcommand{\distrib}{\psi}
\newcommand{\sample}{\leftarrow}
\newcommand{\ffd}[1]{\mathcal{F}_{#1}}
\newcommand{\DLPN}[2]{\code{D}^{\mathsf{LPN}}_{#1,#2}}
\newcommand{\DRLPN}[2]{\code{D}^{\mathsf{RLPN}}_{#1,#2}}
\newcommand{\DMLPN}[2]{\code{D}^{\mathsf{MLPN}}_{#1,#2}}

\DeclareMathOperator{\rot}{\ensuremath{\boldsymbol{\rm rot}}}

\newcommand{\inner}[1]{\langle #1 \rangle}

\newcommand{\OO}{\mathcal{O}}

\newcommand{\F}[1]{\mathbb{F}_{#1}}

\newcommand{\Fq}{\mathbb{F}_q}
\newcommand{\Fqm}{\mathbb{F}_{q^m}}

\newcommand{\CC}{\mathbb{C}}
\newcommand{\QQ}{\mathbb{Q}}
\newcommand{\RR}{\mathbb{R}}
\newcommand{\ZZ}{\mathbb{Z}}
\newcommand{\ideal}[1]{\langle #1 \rangle}

\newcommand{\Gal}{{\rm Gal}}
\newcommand{\Aut}{{\rm Aut}}

\newcommand{\lapin}{\textsc{Lapin}}

\newcommand{\eqdef}{\stackrel{\textrm{def}}{=}}

\renewcommand{\leq}{\leqslant}
\renewcommand{\le}{\leqslant}
\renewcommand{\geq}{\geqslant}
\renewcommand{\ge}{\geqslant}
\newcommand{\map}[4]{\left\{
    \begin{array}{ccc}
      #1 & \longrightarrow & #2\\
      #3 & \longmapsto & #4
    \end{array}\right.
}
\newcommand{\ie}{{\em i.e. }}

\newcommand{\eps}{\varepsilon}
\newcommand{\Prob}{\mathbb{P}}
\newcommand{\adv}{{\rm Adv}}
\newcommand{\fract}[2]{\hbox{\leavevmode
\kern.1em \raise .5ex \hbox{\the\scriptfont0 $#1$}\kern-.1em }/
\hbox{\kern-.15em \lower .25ex \hbox{\the\scriptfont0 $#2$}}
}

\newcommand{\FFDP}{$\mathsf{FF}$--$\mathsf{DP}$}
\newcommand{\MFFDP}{$\mathsf{MFF}$--$\mathsf{DP}$}
\newcommand{\LPN}{$\mathsf{LPN}$}
\newcommand{\RLPN}{$\mathsf{RLPN}$}
\newcommand{\LWE}{$\mathsf{LWE}$}
\newcommand{\MLPN}{$\mathsf{MLPN}$}

\newcommand{\BIKE}{{\sc Bike}}
\newcommand{\HQC}{{\sc HQC}}

 \iftoggle{llncs}{
\author{Maxime Bombar\inst{1,2}\and Alain Couvreur \inst{2,1}\and Thomas Debris-Alazard \inst{2,1}\thanks{This work was funded by the French Agence Nationale de la Recherche through ANR JCJC COLA (ANR-21-CE39-0011).}}

\institute{LIX, CNRS UMR 7161, \'Ecole Polytechnique,\\
  Institut Polytechnique de Paris,\\
  1 rue Honor\'e d'Estienne d'Orves\\
  91120 {\sc Palaiseau Cedex} \and
  Inria\\
  \email{\{maxime.bombar, alain.couvreur, thomas.debris\}@inria.fr} }
}{ \author{Maxime Bombar $^{1,2}$} \email{maxime.bombar@inria.fr}
  \author{Alain Couvreur $^{2, 1}$} \email{alain.couvreur@inria.fr}
  \author{Thomas Debris--Alazard $^{2, 1}$}
  \email{thomas.debris@inria.fr}

  \address{$^{1}$ Laboratoire LIX, \'Ecole Polytechnique, Institut
    Polytechnique de Paris, 1 rue Honor\'e d'Estienne d'Orves, 91120
    Palaiseau Cedex}

  \address{$^{2}$ Inria}

\thanks{This work was funded by the French Agence Nationale de la
  Recherche through ANR JCJC COLA (ANR-21-CE39-0011) and ANR BARRACUDA
  (ANR-21-CE39-0009-BARRACUDA)} }

 \title{On Codes and Learning with Errors over Function Fields}

\begin{document}
\maketitle

 \iftoggle{llncs}{
  \begin{abstract}
  It is a long standing open problem to find search to decision
  reductions for structured versions of the decoding problem of linear
  codes. Such results in the lattice-based setting have been carried
  out using number fields: Polynomial--LWE{}, Ring--\LWE{},
  Module--\LWE{} and so on. We propose a function field version of the
  \LWE{} problem. This new framework leads to another point of view on
  structured codes, {\em e.g.} quasi-cyclic codes, strengthening the
  connection between lattice-based and code-based cryptography. In
  particular, we obtain the first search to decision reduction for
  structured codes. Following the historical constructions in
  lattice--based cryptography, we instantiate our construction with
  function fields analogues of cyclotomic fields, namely {\em Carlitz}
  extensions, leading to search to decision reductions on various
  versions of Ring-\LPN{}, which have applications to secure multi
  party computation and to an authentication protocol.

  \iftoggle{llncs}{
  }
  {
  }
\end{abstract}

   }
   {
    \begin{abstract}
  It is a long standing open problem to find search to decision
  reductions for structured versions of the decoding problem of linear
  codes. Such results in the lattice-based setting have been carried
  out using number fields: Polynomial--LWE{}, Ring--\LWE{},
  Module--\LWE{} and so on. We propose a function field version of the
  \LWE{} problem. This new framework leads to another point of view on
  structured codes, {\em e.g.} quasi-cyclic codes, strengthening the
  connection between lattice-based and code-based cryptography. In
  particular, we obtain the first search to decision reduction for
  structured codes. Following the historical constructions in
  lattice--based cryptography, we instantiate our construction with
  function fields analogues of cyclotomic fields, namely {\em Carlitz}
  extensions, leading to search to decision reductions on various
  versions of Ring-\LPN{}, which have applications to secure multi
  party computation and to an authentication protocol.

  \iftoggle{llncs}{
  }
  {
  }
\end{abstract}

    \setcounter{tocdepth}{1}
    \tableofcontents
   }

\section{Introduction}
\label{sec:intro}

{\bf \noindent Code-based cryptography.} Error correcting codes are
well known to provide quantum resistant cryptographic primitives such
as authentication protocols \cite{S93,HKLPP12}, signatures
\cite{CFS01,DST19a} or encryption schemes such as McEliece \cite{M78}.
These code-based cryptosystems were built to rely on the following
hard problem: finding a close (or far away) codeword to a given word,
a task called {\em decoding}. In the case of random linear codes of
length $n$, which is the standard case, this problem can be expressed
as follows. First, we are given a vector space $\code{C}$ ({\em i.e.}
the code) of $\Fq^{n}$ generated by the rows of some random matrix
$\vec{G}\in\Fq^{k\times n}$, namely:
\begin{equation}\label{eq:def_code}
	\code{C} \eqdef \{\vec{m}\vec{G} \mid \vec{m} \in \Fq^k\}.
\end{equation}
The decoding problem corresponds, given $\vec{G}$ (in other words
$\code{C}$) and some noisy codeword $\vec{m}\vec{G} + \vec{e}$ where
the number of non-zero coordinates of $\vec{e}$ is equal to $t$ (its
Hamming weight is $|\vec{e}| = t$), to find the error $\vec{e}$ or
what amounts to the same, the original codeword $\vec{m}\vec{G}$.

Usually this decoding problem is considered in the regime where the
code rate $R \eqdef \frac{k}{n}$ is fixed, but there are also other
interesting parameters for cryptographic applications. For instance,
the Learning Parity with Noise problem (\LPN{}) corresponds to the
decoding problem where $n$ is the number of samples, $k$ the length of
the secret while the error is sampled according to a Bernoulli
distribution of fixed rate $t/n$. As the number of samples in \LPN{} is
unlimited, this problem actually corresponds to decoding a random code
of rate arbitrarily close to $0$.

\iftoggle{llncs}{}{While the security of many code-based cryptosystems relies on the
hardness of the decoding problem, it can also be based on finding a
``short'' codeword for the Hamming metric (as in \cite{MTSB12} or in
\cite{AHIKV17,BLVW19,YZWGL19} to build collision resistant hash
functions). It turns out that decoding and finding short codewords are
closely related. It has been shown in \cite{DRT21}
(following the original quantum reduction of Regev in the lattice/\LWE{}
case \cite{R05}) that decoding some code $\code{C}$ is
quantumly-harder than finding a short codeword in its dual
$\code{C}^{\perp}$ (for the standard inner product in $\Fq^{n}$).
A reduction from decoding to the problem of finding short
codewords is also known but in an \LPN{} context
\cite{AHIKV17,BLVW19,YZWGL19}.}

Despite the promising approach of McEliece, there are two drawbacks if one
follows it to design a cryptosystem.
First, the public data in McEliece is a representation of a code which
has to look like random. Assuming this pseudo-randomness property, the security relies on the hardness of the
decoding problem. In that case one needs to publish $\Omega(n^{2})$
bits but at the same time, best generic decoding algorithms
have a complexity exponential in the number $t$ of errors to correct. Therefore, to reach
a security level of $2^{\lambda}$, the public data are of order $\Theta(\lambda^{2})$ if
$t = \Theta(n)$ or even worse of the order $\Theta(\lambda^{4})$ if
$t = \Theta\left(\sqrt{n}\right)$. On the other hand, in
McEliece-like cryptosystems,
the owner of the secret key has to know an efficient decoding
algorithm for the public code.
It turns out that codes for which we
know an efficient decoding algorithm are obtained via polynomial evaluations ({\em e.g.} Goppa codes) or short vectors ({\em e.g.} MDPC codes).
Thus, the
owner of the secret key has to hide the peculiar description of the code he
publishes. It leads to the fact that in McEliece-like cryptosystems,
the security also relies on the difficulty to distinguish the code that
is made public from a random one. This is a second assumption to make in addition to
the hardness of the decoding problem.
\newline

{\bf \noindent Alekhnovich cryptosystem.} In 2003, Alekhnovich \cite{A03}
introduced a new approach to design an encryption scheme
based on error correcting codes. Unlike McEliece cryptosystem,
Alekhnovich truly
relies on the hardness of decoding random codes.
It starts from a random code
$\code{C}$ and proceeds as follows:

\begin{itemize}
  \item \textit{Key Generation.} Let $\vec{e}_{\mathsf{sk}}\in\F2^{n}$
        of small Hamming weight. The public key is
        $(\code{C}, \vec{c} + \vec{e}_{\mathsf{sk}})$ where
        $\vec{c}\in \code{C}$ and the secret key is
        $\bfe_{\mathsf{sk}}$.

        \vspace{\baselineskip}

  \item \textit{Encryption.} To encrypt one bit $\beta\in\{0, 1\}$
        set:

        \begin{itemize}
          \item $\Enc(1) \eqdef \bfu$ where $\bfu\in \F2^{n}$ is a uniformly
                random vector.
          \item $\Enc(0) \eqdef \vec{c}^{\ast}+\vec{e}$ where
                $\vec{e}$ is of small Hamming weight and
                $\vec{c}^{\ast}$ lies in the dual of the code
                $\code{C}_{\mathsf{pub}}$ spanned by $\code{C}$ and
                $\vec{c}+\vec{e}_{\mathsf{sk}}$.
        \end{itemize}

        \vspace{\baselineskip}

      \item \textit{Decryption.} The decryption of $\Enc(\beta)$ is
        $\inner{\Enc(\beta), \vec{e}_{\mathsf{sk}}}$, where
        $\inner{\cdot,\cdot}$ is the usual inner product on $\F2^{n}$.
\end{itemize}
The correction of this procedure relies on the fact that
\[
  \inner{\Enc(0), \vec{e}_{\mathsf{sk}}}
  = \inner{\vec{c}^{\ast} + \vec{e}, \vec{e}_{\mathsf{sk}}}
  = \inner{\vec{e}, \vec{e}_{\mathsf{sk}}},
\]
where we used that
$\vec{e}_{\mathsf{sk}}\in \code{C}_{\mathsf{pub}}$ while
$\vec{c}^{\ast}$ lies in its dual. Now, this inner product is equal to
$0$ with overwhelming probability as $\vec{e}_{\mathsf{sk}}$ and
$\vec{e}$ are of small Hamming weight. On the other hand,
$\inner{\Enc(1), \vec{e}_{\mathsf{sk}}}$ is a uniformly random bit.

Therefore, contrary to McEliece cryptosystem, the security of
Alekhnovich scheme does not depend on hiding the description of a code:

\begin{itemize}[label=\textbullet]
  \item \textit{Key security.} Recovering the private key from public
        data amounts to decoding the random code $\code{C}$\iftoggle{llncs}{.}{, or
        finding a short vector in the code spanned by $\code{C}$ and
        $\vec{c} + \vec{e}_{\mathsf{sk}}$.}
  \item \textit{Message security.} Recovering the plaintext from the
        ciphertext is tantamount to \textit{distinguishing} a noisy
        codeword from a uniformly random vector.
\end{itemize}

\noindent The message security relies on the \textit{decision}
version of the decoding problem. Search and decision versions of the
decoding problem are known to be computationally equivalent using Goldreich-Levin theorem
\cite{FS96}. However, Alekhnovich cryptosystem suffers from major
drawbacks:
\begin{enumerate}
  \item Encrypting one bit amounts to sending $n$ bits;
  \item The public key size is quadratic in the length of ciphertexts.
\end{enumerate}

\noindent While the first issue can easily be addressed, the second flaw needs
more work, and as is, Alekhnovich cryptosystem is not practical.
However, the approach itself was a major breakthrough in code-based
cryptography. It was inspired by the work of Ajtai and Dwork
\cite{AD97} whose cryptosystem is based on solving hard lattice
problems. The latter reference from Ajtai and Dwork is also the
inspiration of Regev famous Learning With Errors (\LWE{}) problem
\cite{R05}, which is at the origin of an impressive line of work. As
Alekhnovich cryptosystem, the original \LWE{} cryptosystem was not
practical either and, to address this issue, structured versions
were proposed, for instance Polynomial-\LWE{} \cite{SSTX09}, Ring--\LWE{}
\cite{LPR10}, Module--\LWE{} \cite{LS15}.
\newline

{\bf \noindent Structured decoding problem.} In the same fashion, for
code--based public key encryptions, it has been proposed to restrict
to codes that can be represented more compactly to reduce the key
sizes.  In McEliece setting, the story begins in 2005 with the results
of \cite{G05} that suggest to use $\ell$--quasi-cyclic codes, \ie
codes
that are generated by a matrix $\vec{G}$ formed out of $\ell$
blocks:
\begin{equation}\label{eq:qc-code}
  \bfG = \begin{pmatrix}
    \rot\left(\bfa^{(1)}\right) & \cdots & \rot\left(\bfa^{(\ell)}\right)
\end{pmatrix},
\end{equation}
each block being a circulant matrix, \ie of the form
\[
  \rot(\bfa) \eqdef \begin{pmatrix}
    a_{0} & a_{1} & \dots & \dots & a_{k-1} \\
    a_{k-1} & a_{0} & \dots & \dots & a_{k-2} \\
    \vdots & \ddots & \ddots & & \vdots \\
    \vdots & & \ddots & \ddots & \vdots \\
    a_{1} & a_{2} & \dots & a_{k-1} & a_{0}
  \end{pmatrix} \text{ with }\bfa \in \Fq^{k}.
\]
The key point is that such codes have a large automorphism group $G$,
and instead of publishing a whole basis, one can only publish a
generating set for the $\Fq[G]$--module structure of the code. That is
to say, a family of vectors whose orbit under the action of $G$ spans
the code. For instance, in the case of quasi-cyclic codes
(\ref{eq:qc-code}), one can publish only the first row of the
$\ell$-circulant generator matrix. It can be argued that the
  quasi--cyclicity could be used to improve the speed-up of generic
  decoding, but the best known approach in the generic case uses
DOOM \cite{S11} which allows to divide
the complexity of decoding by at most $\sqrt{\#G}$, the latter complexity
remaining exponential with the same exponent. Hence, one can keep the
same security parameter, while the size of the public key can be
divided by a factor $O(\#G)$.

This idea leads to very efficient encryption schemes such as \BIKE{}
\cite{AABBBBDGGGMMPSTZ21}, in the McEliece fashion, or \HQC{}
\cite{AABBBDGPZ21a} which is closer to Ring--\LWE{}. Both proposals use
2-quasi-cyclic codes and have been
selected to the third round of NIST competition as alternate
candidates. Other structured variants of the decoding problem
(referred to as Ring--\LPN{}) were also proposed with
  applications to authentication \cite{HKLPP12} or secure MPC
  \cite{BCGIKS20}.\iftoggle{llncs}{ }{ Note that the idea to use codes
  equipped with a non trivial ring action has also been used in rank
  metric \cite{ABDGHRTZABBBO19,AABBBDGZCH19}.}

In other words, the security of those cryptosystems now rely on some
structured variant of the decoding problem.
\newline

{\bf \noindent A Polynomial representation.}
It turns out that a convenient way of seeing $\ell$-quasi-cyclic
codes, is to represent blocks of their generator matrix as elements of
the quotient ring ${\Fq[X]}/{(X^{n}-1)}$, via the $\Fq$--isomorphism:
\[
  \map{\Fq^n}{\fract{\Fq[X]}{(X^{n}-1)}}{\bfa \eqdef (a_{0}, \dots, a_{n-1})}{
    \bfa(X) \eqdef \displaystyle\sum_{i=0}^{n-1}a_{i}X^{i}.}
\]
A simple computation shows that the product of two elements of
${\Fq [X]}/{(X^{n}-1)}$ can be represented with the operator $\rot(\cdot)$:
\[
  \bfu(X)\bfv(X) \mod (X^{n}-1)
  = \bfu\cdot\rot(\bfv)
  = \bfv \cdot  \rot(\bfu)
  = \bfv(X) \bfu(X) \mod (X^{n}-1).
\]
From now on, $\bfu$ can denote either a vector of $\Fq^{n}$ or a
polynomial in ${\Fq[X]}/{(X^{n}-1)}$, and the product of two elements
$\bfu\bfv$ is defined as above.

Consider an $\ell$-quasi-cyclic code with a generator matrix $\bfG$ in
$\ell$-circulant form. Let $\bfs\in\Fq^{n}$ be a secret word of the
ambient space and let $\bfe\in\Fq^{\ell n}$ be an error vector. Under
the above map, the noisy codeword $\bfs\bfG + \bfe$ is
represented by $\ell$ samples of the form
$\bfs\bfa^{(j)} + \bfe^{(j)} \in {\Fq[X]}/{(X^{n}-1)}$ and the
decoding problem of $\ell$-circulant codes corresponds to recovering
the secret $\bfs$ given $\ell$ samples. This can be seen as a code
analogue of the Ring--\LWE{} problem, with access to a fixed number of
samples $\ell$. The rate of the code is $\frac{1}{\ell}$, so
increasing the number of samples corresponds to decode a code
whose rate goes to $0$.

A natural generalization would be to consider multiple rows of
circulant blocks. In this situation, the generator matrix $\bfG$ is of the form
\[
  \bfG = \begin{pmatrix}
    \rot(\bfa^{(1, 1)}) & \cdots & \rot(\bfa^{(1,\ell)})\\
    \vdots & & \vdots\\
    \rot(\bfa^{(m,1)}) & \cdots & \rot(\bfa^{(m,\ell)})\\
\end{pmatrix}
\]
and a noisy codeword $\bfs\bfG + \bfe$ is now represented by
$\ell$ samples of the form
\[
  \sum_{i=1}^{m}\bfs_{i}\bfa^{(i, j)} + \bfe_{j} \in \fract{\Fq[X]}{(X^{n}-1)}
\]
where $\bfs$ can be considered as a collection of $m$ secrets
$\bfs_{1}, \dots, \bfs_{m}$. This would be the code analogue of
Module--\LWE{}, with a rank $m$ module and $\ell$ samples, introduced in
\cite{LS15}.

Contrary to structured lattice cryptosystems, up to now, no reduction
from the search to the decision version of the structured
decoding problem was known. This was pointed out by
NIST \cite{NIST20}, and was a reason for those code-based
cryptosystems to be only considered as alternate candidates for the
third round.  Actually even before NIST standardization process, this
lack of search to decision reduction was already pointed out by the
authors of the Ring--\LPN{} based authentication scheme \lapin{}
\cite{HKLPP12}.  \newline

{\bf \noindent Our contribution.} To handle this lack of
search to decision reduction in the code setting, we propose in this
article a new generic problem called \FFDP{}, for {\em Function Field
  Decoding Problem}, in the Ring--\LWE{} fashion. One of the key ideas
consists in using function fields instead of number fields, the latter
being used in the lattice case. This framework enables us to adapt
directly the search to decision reduction of \cite{LPR10} in the case
of codes.
Frequently in the literature on Ring--\LWE{}, the search to decision
reduction is instantiated with cyclotomic number fields. In the same
spirit we present an instantiation with function fields analogues of
cyclotomic fields, namely the so-called {\em Carlitz extensions}. As
we show, this framework is for instance enough to provide a search to
decision reduction useful in the context of \lapin{} \cite{HKLPP12} or
for a $q$--ary analogue of Ring--\LPN{} used for secure multiparty
computation \cite{BCGIKS20}.
If our reduction does not work for every schemes based on structured
codes such as \HQC{}, we believe that our work paves the way towards a
full reduction.

\begin{remark}
	Note that the use of function fields in coding theory is far from
  being new. Since the early 80's and the seminal work of Goppa
  \cite{G81}, it is well--known that codes called {\em Algebraic
    Geometry} (AG) codes can be constructed from algebraic curves or
  equivalently from function fields and that some of these codes have
  better asymptotic parameters than random ones \cite{TVZ82}. However,
  the way they are used in the present work is completely different.
  Indeed, AG codes are a natural generalization of Reed--Solomon and,
  in particular, are codes benefiting from efficient decoding
  algorithms (see for instance surveys \cite{HP95,BP08,CR21}). In the
  present article, the approach is somehow orthogonal to the AG codes
  setting since we use function fields in order to introduce generic
  problems related to structured codes for which the decoding problem
  is supposed to be hard.
\end{remark}

 {\bf \noindent A function field approach.}
Lattice-based cryptography has a long standing history of using number
fields and their rings of integers to add some structure and reduce
the key sizes. Recall that number fields are algebraic extensions of
$\QQ$ of the form
\[
  K \eqdef \fract{\QQ[X]}{(f(X))},
\]
where $f$ is an irreducible polynomial, and the ring of integers
$\OO_{K}$ is the integral closure of $\ZZ$ in $K$, \ie it is the
subring of $K$ composed of elements which are roots of monic
polynomials with coefficients in $\ZZ$. For instance, cyclotomic
extensions are of the form
$K = \QQ(\zeta_{m}) = {\QQ[X]}/{(\Phi_{m}(X))}$ where $\zeta_{m}$
is a primitive $m$-th root of unity and $\Phi_{m}$ is the $m$-th
cyclotomic polynomial. The ring of cyclotomic integers has a very
specific form, namely $\OO_{K} = \ZZ[\zeta_{m}]$. One of the most used
case is when $m$ is a power of $2$. In this case, setting $m = 2n$, we
have $\Phi_{m} = \Phi_{2n} = X^{n}+1$ and
$\OO_{K} = {\ZZ[X]}/{(X^{n}+1)}$. Such rings have been widely
used since they benefit from a very fast arithmetic thanks to the fast
Fourier transform. In the Ring--\LWE{} setting, one reduces all the
samples modulo a large prime element $q\in\ZZ$ called the {\em
  modulus} and hence considers the ring
${(\ZZ/q\ZZ)[X]}/{(X^{n}+1)}$.\iftoggle{llncs}{}{ Due to inherent considerations of
the Euclidean metric, errors are drawn according to a {\em continuous}
distribution (e.g a Gaussian distribution) $\chi$ over the Euclidean space
$K\otimes_{\QQ}\RR = {\RR[X]}/{(X^{n}+1)}$
and one has to
introduce a technical tool called {\em smoothing parameter} to handle
the {\em discrete} error distributions used in practice. It should be
noted that an equivalent of the smoothing parameter will not be
necessary in our case because our error model will remain discrete.}

When moving from structured lattices to structured codes, it would be
tantalizing to consider the ring ${\Fq[X]}/{(X^{n}-1)}$ as the analogue of
${\ZZ[X]}/{(X^{n}+1)}$.
However, if the two rings have a similar expression they have a
fundamental difference. Note for instance that the former is finite
while the latter is infinite. From a more algebraic point of view,
${\Fq[X]}/{(X^{n}-1)}$ is said to have {\em Krull dimension} $0$
while ${\ZZ[X]}/{(X^{n}+1)}$ has {\em Krull dimension $1$}. In
particular, the former has only a finite number of ideals while the
latter has infinitely many prime ideals.
The main idea of the present article is to lift the decoding problem
and to see ${\Fq[X]}/{(X^{n}-1)}$ as a quotient ${R}/{I}$ of
some ring $R$ of Krull dimension $1$. The ideal $I$ will be the
analogue of the {\em modulus}. This setting can be achieved using so-called
{\em function fields}. It could be argued
that the results of this article could have been obtained without
introducing function fields. However, we claim that
function fields are crucial for at least three
reasons:
\begin{enumerate}
  \item Introducing function fields permits to establish a strong
        connection between cryptography based on structured lattices
        involving number fields on the one hand and cryptography based
        on structured codes on the other hand.
      \item Number theory has a rich history with almost one
          hundred years of development of the theory of function
          fields. We expect that, as number fields did for
        structured lattices, function fields will yield a remarkable
        toolbox to study structured codes and cryptographic questions
        related to them.
      \item A third and more technical evidence is that a crucial part
        of the search to decision reduction involves some Galois
        action. We claim that, even if for a specific instantiation,
        this group action could have been described in a pedestrian
        way on the finite ring ${\Fq [X]}/{(X^n - 1)}$, without
        knowing the context of function fields, such a group action
        would really look like ``a rabbit pulled out of a hat''. In
        short, this group action, which is crucial to conclude the
        search to decision reduction, cannot appear to be something
        natural without considering function fields.
\end{enumerate}

It is well--known for a long time that there is a noticeable analogy
between the theory of number fields and that of function fields.
Starting from the ground, the rings $\ZZ$ and $\Fq[T]$ share a lot of
common features. For instance, they both have an Euclidean division.
Now if one considers their respective fraction fields $\QQ$ and
$\Fq(T)$, finite extensions of $\QQ$ yield the number fields while
finite separable extensions of $\Fq(T)$ are called {\em function
  fields} because they are also the fields of rational functions on
curves over finite fields. Now, a similar arithmetic theory can be
developed for both with rings of integers, orders, places and so on.
Both rings of integers are {\em Dedekind domains}. In particular,
every ideal factorizes uniquely into a product of prime ideals, and
the quotient by any non-zero ideal is always finite. A dictionary
summarizing this analogy between number fields and function fields is
represented in Table \ref{fig:ff_analogy}. Note that actually, many
properties that are known for function fields are only conjectures for
number fields. The best example is probably the Riemann hypothesis
which has been proved by Weil in the early 1940s in the function field
case.

\begin{table}[ht]
  \centering
\[
  \begin{array}{|c|c|}
    \hline
    \text{ Number fields } & \text{ Function fields } \\
    \hline
    \QQ & \Fq(T) \\
    \ZZ & \Fq[T] \\
    \text{Prime numbers } q \in \ZZ & \text{Irreducible polynomials } Q \in \Fq[T]\\
        & \\
    K = \fract{\QQ[X]}{(f(X))} & K = \fract{\Fq(T)[X]}{(f(T, X))}\\
        & \\
    \makecell{\OO_{K} \\ = \text{Integral closure of $\ZZ$} \\ \text{\emph{Dedekind} domain}}  & \makecell{\OO_{K} \\ = \text{Integral closure of $\Fq[T]$} \\ \text{\emph{Dedekind} domain}} \\
        & \\
    { \textbf{characteristic 0}} & {\bf \textbf{characteristic} >0}\\
    \hline
  \end{array}
\]
\caption{A Number-Function fields analogy}
\label{fig:ff_analogy}
\end{table}

With this analogy in hand, the idea is to find a nice function field
$K$ with ring of integers $\OO_{K}$ and an irreducible polynomial
$Q\in\Fq[T]$, called the {\em modulus}, such that
${\OO_{K}}/{Q\OO_{K}} = {\Fq[X]}/{(X^{n}-1)}$. Following the path of
\cite{LPR10}, we are able to provide a search to decision reduction
for our generic problem \FFDP{} when three conditions hold:

\begin{enumerate}
  \item\label{hyp:galois} The function field $K$ is Galois.
  \item\label{hyp:splitting_modulus} The modulus $Q$ does not ramify
        in $\OO_{K}$, meaning that the ideal $Q\OO_{K}$
        factorizes in product of distinct prime ideals.
  \item\label{hyp:galois_distribution} The distribution of errors is
        invariant under the action of the Galois group.
\end{enumerate}

This framework is enough to provide a search to decision reduction
useful in the context of \lapin{} \cite{HKLPP12} or for a $q$--ary
analogue of Ring--\LPN{} used for secure MPC \cite{BCGIKS20}. It should be emphasized that, in the case of
\lapin{}, the search to decision reduction requires to adapt the
definition of the noise which will remain built by applying
independent Bernouilli variables but with a peculiar choice of
$\F{2}$--basis of the underlying ring ${\F{2}[X]}/{(f(X))}$. The
chosen basis is a {\em normal} basis, \ie is globally invariant with
respect to the Galois action. This change of basis is very similar to
the one performed in lattice based-cryptography when, instead of
considering the monomial basis $1, X, \dots, X^{n-1}$ in an order
${\ZZ[X]}/{(f(X))}$, one considers the canonical basis after
applying the Minkowski embedding.  Indeed, the latter is Galois invariant.
We emphasize that, here again, the
function field point of view brings in a Galois action which cannot
appear when only considering a ring such as
${\F{2}[X]}/{(f(X))}$. This is another evidence of the need for
introducing function fields.
\newline

{\noindent \bf Outline of the article.} The present article is organised as follows. Section \ref{sec:prereqFF}  recalls the necessary background about function fields (definitions and important properties). In Section \ref{sec:FFDP} we present the \FFDP{} problem (search and decision versions) as well as our main theorem (Theorem \ref{thm:main}) which states the search to decision reduction in the function field setting.\iftoggle{llncs}{ A proof of this theorem is given in Supplementary Materials for sake of completeness.}{ A proof of this theorem is given in Section \ref{sec:StD}. A reader only interested about the framework of functions fields and our instantiations can safely skip this section.} In Section \ref{sec:carlitz} we give a self contain presentation of Carlitz extensions. They will be used to instantiate our search to decision reduction in Section \ref{sec:applications}, which provides our applications.

 \section{Prerequisites on function fields}\label{sec:prereqFF}
In this section, we list the minimal basic notions on the arithmetic of
function fields that are needed in the sequel. A dictionary drawing
the analogies has been given in Table~\ref{fig:ff_analogy}. For
further references on the arithmetic of function fields, we refer the
reader to \cite{S09,R02}.

Starting from a finite field $\Fq$, {\em a function field} is
a finite extension $K$ of $\Fq(T)$ of degree $n>0$ of the form
\[
  K = \fract{\Fq(T)[X]}{(P(T, X))}
\]
where $P(T,X) \in \Fq(T)[X]$ is irreducible of degree $n$. The field
$K \cap \overline{\F{}}_{q}$ is referred to as {\em the field of
  constants} or {\em constant field} of $K$, where
$\overline{\F{}}_q$ is the algebraic closure of $\Fq$.
In the sequel, we will assume that
$\Fq$ is the full field of constants of $K$,
which is equivalent for $P(T, X)$ to be
irreducible even regarded as a an element of
$\overline{\F{}}_q(T)[X]$ (\cite[Cor.~3.6.8]{S09}).

Similarly to the number field case, one can define the ring of
integers $\OO_K$ as the the ring of elements of $K$ which are the roots
of a monic polynomial in $\Fq[T][X]$. This ring is a {\em Dedekind domain}.
In particular, any ideal $\mathfrak{P}$ has a unique decomposition
$\mathfrak{P}_1^{e_1}\cdots \mathfrak{P}_r^{e_r}$ where the
$\mathfrak{P}_i$'s are prime ideals.

In the sequel, we frequently focus on the following
setting represented in the diagram below: starting from a prime ideal
$\mathfrak{p}$ of $\Fq[T]$ (which is nothing but the ideal generated
by an irreducible polynomial $Q(T)$ of $\Fq[T]$), we consider the ideal
$\mathfrak P \eqdef \mathfrak p \OO_K$ and its decomposition:
\[
  \mathfrak{P} = \mathfrak{P}_1^{e_1} \cdots \mathfrak{P}_r^{e_r}.
\]
\begin{center}
  \begin{tikzpicture} \matrix (m) [matrix of math nodes,row
    sep=3em,column sep=4em,minimum width=2em] {
      \mathfrak{P} \subset \OO_K & K \\
      \mathfrak{p} \subset \Fq[T] & \Fq(T) \\};

    \path[-] (m-1-1) edge (m-2-1) edge (m-1-2) (m-2-1) edge (m-2-2)
    (m-1-2) edge (m-2-2);
\end{tikzpicture}
\end{center}
The prime ideals $\mathfrak{P}_{i}$'s are said to {\em lie above}
$\mathfrak{p}$. The exponents $e_i$'s are referred to as the {\em
  ramification indexes}, and the extension is said to be {\em
  unramified} at $\mathfrak{P}$ when all the $e_i$'s are equal to $1$.
Another important constant related to a $\mathfrak{P}_i$ is its {\em
  inertia degree}, which is defined as the extension degree
$f_i \eqdef [\OO_K/\mathfrak{P}_i:\Fq[T]/\mathfrak p]$ (one can prove
that $\OO_K/\mathfrak{P}_i$ and $\Fq[T]/\mathfrak p$ are both
finite fields). The Chinese Remainder Theorem (CRT) induces a ring
isomorphism between $\OO_{K}/\mathfrak{P}$ and
$\prod_{i=1}^{r}\OO_{K}/\mathfrak{P}_{i}^{e_{i}}$. In particular, when
the extension is unramified at $\mathfrak{P}$, the quotient
$\OO_{K}/\mathfrak{P}$ is a product of finite fields.
Finally, a well-known result asserts that
\begin{equation}\label{eq:fundamental_equality}
  n = [K:\Fq(T)] = \sum_{i=1}^r e_i f_i.
\end{equation}

\noindent {\bf Finite Galois extensions.} \iftoggle{llncs}{}{Recall that a finite
algebraic field extension $L/K$ is said to be a {\em Galois extension}
when the automorphism group
\[
  \Aut(L/K) \eqdef \{\sigma \colon L \rightarrow L \mid \sigma \text{
    is an isomorphism with } \sigma(a) = a \text{ for all } a\in K\}
\]
has cardinality $[L:K]$. In that case, we refer to $\Aut(L/K)$ as the
{\em Galois group} of $L/K$ and write $\Gal(L/K) \eqdef \Aut(L/K)$.
Galois extensions whose Galois group is abelian are called {\em
  abelian extensions}. Galois extensions have many properties that do
not hold in general field extensions.

When $L/K$ is a Galois extension, and if $H$ is a subgroup of
$G\eqdef \Gal(L/K)$, then the set
\[
  L^{H}\eqdef \{ a \in L \mid \sigma(a) = a \text{ for all } \sigma\in H \}
\]
is a field called the {\em fixed field} of $H$. By definition
$L^{G} = K$. Furthermore, the extension $L/L^{H}$ is Galois with Galois
group $H$. On the other hand, the extension $L^{H}/K$ may not be
Galois in general, but it is the case when $H$ is a normal subgroup of
$G$, and $\Gal(L^{H}/L) = G/H$. This is particularly true when $L$ is
an abelian extension.} Consider $K/\Fq(T)$ a Galois function field (\ie a function field $K$
which is a Galois extension of $\Fq(T)$), with Galois group
$G\eqdef \Gal(K/\Fq(T))$. Then, $G$ keeps $\OO_{K}$ globally
invariant.  Furthermore, given $\mathfrak{p}$ a prime ideal of
$\Fq[T]$, the group $G$ acts transitively on the set
$\{\gothP_{1},\dots, \gothP_{r}\}$ of prime ideals of $\OO_{K}$ lying
above $\mathfrak{p}$: for any $i\neq j$ there exists
$\sigma\in\Gal(K/\Fq(T))$ such that
$\sigma(\gothP_{i}) = \gothP_{j}$. In particular, all the
ramification indexes $e_{i}$ ({\em resp.} the inertia degrees $f_{i}$)
are equal and denoted by $e$ ({\em resp.} $f$):
$\gothP \eqdef \mathfrak{p}\OO_{K} = (\gothP_{1}\dots\gothP_{r})^{e}$
and \eqref{eq:fundamental_equality} becomes $n = efr$.  Another
consequence which will be crucial for the applications, is that the
action of $G$ on $\OO_{K}$ is well--defined on
$\OO_{K}/\gothP$ and simply permutes factors $\OO_{K}/\gothP_{i}^{e}$.
The {\em decomposition group} of $\gothP_{i}$ over $\gothp$ is
\[
  D_{\gothP_{i}/\mathfrak{p}}\eqdef \{\sigma \in G \mid \sigma
  \left(\gothP_{i}\right) = \gothP_{i}\}.
\]
  It has cardinality
$e\times f$. In particular, when $K$ is unramified at $\gothP$, the
field $\OO_{K}/\gothP_{i}$ is $\mathbb{F}_{q^{f}}$ and the action of
$D_{\gothP_{i}/\mathfrak{p}}$ on it is the Frobenius automorphism: the
reduction modulo $\gothP_{i}$ yields
an isomorphism
\begin{equation}\label{eq:decomp_isom}
  D_{\gothP_{i}/\mathfrak{p}} \simeq \Gal(\mathbb{F}_{q^{f}}/\Fq).
\end{equation}
Finally, all the decomposition groups of primes above $\gothp$ are
conjugate: for any $i\neq j$ there exists $\sigma\in G$ such that
$D_{\gothP_{i}/\mathfrak{p}} = \sigma D_{\gothP_{j}/\mathfrak{p}}\sigma^{-1}$.

 \section{A function field approach for search to decision reduction}
\label{sec:FFDP}

{\bf \noindent Search and decision problems.}  In this section, we
introduce a new generic problem that we call \FFDP{}, which is the
analogue of Ring--\LWE{} in the context of function fields. Then, we
give our main theorem which states the search-to-decision reduction of
\FFDP{}. Since function fields and number fields share many
properties, the present search to decision reduction, that is proven
in Section \ref{sec:StD}, will work similarly as in \cite{LPR10}.

Consider a function field ${K}/{\Fq(T)}$ with constant field
$\Fq$ and ring of integers $\OO_{K}$ and let $Q(T)\in\Fq[T]$. Let
$\gothP \eqdef Q\OO_{K}$ be the ideal of $\OO_{K}$ generated by $Q$.
Recall that ${\OO_{K}}/{\gothP}$ is a finite set. \FFDP{} is parameterized
by an element $\bfs\in{\OO_{K}}/{\gothP}$ called the {\em secret} and
$\distrib$ be a probability distribution over ${\OO_{K}}/{\gothP}$ called the
{\em error distribution}.

\begin{definition}
[\FFDP{} Distribution]\label{def:FFDP_distribution}
  A sample $(\vec{a}, \vec{b}) \in {\OO_{K}}/{\gothP}\times {\OO_{K}}/{\gothP}$ is
  distributed according to the \FFDP{} distribution modulo $\gothP$
  with secret $\bfs$ and error distribution $\distrib$ if
  \begin{itemize}[label=\textbullet]
\item $\bfa$ is uniformly distributed over ${\OO_{K}}/{\gothP}$,
  \item $\bfb = \bfa \bfs + \bfe \in {\OO_{K}}/{\gothP}$ where $\vec{e}$ is distributed according to $\distrib$.
\end{itemize}
A sample drawn according to this distribution will be denoted by
$(\bfa, \bfb)\sample \ffd{\bfs, \distrib}$.
\end{definition}

The aim of the search version of the \FFDP{} problem is to recover the secret
$\bfs$ given samples drawn from $\ffd{\bfs, \distrib}$. This is
formalized in the following problem.

\iftoggle{llncs}{\begin{definition}}{\begin{problem}}
[\FFDP{}, Search version]
  \label{pb:FFDP_search} Let $\bfs\in {\OO_{K}}/{\gothP}$, and let
  $\distrib$ be a probability distribution over ${\OO_{K}}/{\gothP}$. An
  instance of \FFDP{} problem consists in an oracle giving
  access to independent samples
  $(\bfa, \bfb) \sample \ffd{\vec{s}, \distrib}$. The goal is to recover
  $\bfs$.
\iftoggle{llncs}{\end{definition}}{\end{problem}}

\begin{remark}
  This problem should be related to structured versions of the
  decoding problem. Indeed, recall from the discussion in the
  introduction that, using the polynomial representation, the decoding
  problem of random quasi-cyclic codes corresponds to recovering a
  secret polynomial $\bfs(X)\in {\Fq[X]}/{(X^{n}-1)}$ given access
  to samples of the form
  $\bfa\bfs + \bfe \in {\Fq[X]}/{(X^{n}-1)}$ where $\bfa$ is
  uniformly distributed in ${\Fq[X]}/{(X^{n}-1)}$. This can be
  rephrased within the \FFDP{} framework as follows. Consider the
  polynomial $f(T, X) \eqdef X^{n}+T-1 \in \Fq(T)[X]$. When $n$ is not
  divisible by the characteristic of $\Fq$, $f$ is a separable
  polynomial. Moreover, by Eisenstein criterion $f$ is
  irreducible. Define the function field $K$ generated by $f$, namely
  the extension $K\eqdef {\Fq(T)[X]}/{(f(T, X))}$. One
  can prove that $\OO_{K}$ is exactly ${\Fq[T][X]}/{(f(T, X))}$.
Now, let $\mathfrak{p}$ be the ideal of
$\Fq[T]$ defined by the irreducible polynomial $T$, and let
$\mathfrak{P}\eqdef \mathfrak{p}\OO_{K} = T\OO_{K}$ be the
corresponding ideal of $\OO_{K}$. Then the following isomorphisms hold
\[
  \fract{\OO_{K}}{\mathfrak{P}} \simeq \fract{\Fq[T, X]}{(T, X^{n}+T-1)} \simeq
  \fract{\Fq[X]}{(X^{n}-1)}.
\]
With this particular instantiation, ${\OO_{K}}/{\mathfrak{P}}$ is
exactly the ambient space from which the samples are defined in the
structured versions of the decoding problem.
As
a consequence, \FFDP{} is a generalization of structured versions of
the decoding problem, when
considering arbitrary function fields and ideals.
\end{remark}

For cryptographic applications, we are also interested in the {\em
  decision} version of this problem. The goal is now to distinguish
between the \FFDP{} distribution and the uniform distribution over
${\OO_{K}}/{\gothP} \times {\OO_{K}}/{\gothP}$.

\iftoggle{llncs}{\begin{definition}}{\begin{problem}}
		[\FFDP{}, Decision version]
    \label{pb:FFDP_decision}
    Let $\bfs$ be drawn uniformly at random in ${\OO_{K}}/{\gothP}$ and let
    $\distrib$ be a probability distribution over ${\OO_{K}}/{\gothP}$.
    Define $\code{D}_{0}$ to be the uniform distribution over
    ${\OO_{K}}/{\gothP} \times {\OO_{K}}/{\gothP}$, and $\code{D}_{1}$ to be the
    \FFDP{} distribution with secret $\bfs$ and error distribution
    $\distrib$. Furthermore, let $b$ be a uniform element of $\{0,1\}$.
    Given access to an oracle $\code{O}_{b}$ providing
    samples from distribution $\code{D}_{b}$, the goal of the decision \FFDP{}
    is to recover $b$.
\iftoggle{llncs}{\end{definition}}{\end{problem}}

\begin{remark}
    For some applications, for instance to MPC, it is more convenient
    to have the secret $\bfs$ drawn from the error distribution $\psi$
    instead of the uniform distribution over $\OO_{K}/\gothP$. In the
    lattice-based setting, this version is sometimes called \LWE{}
    with {\em short secret} or \LWE{} in {\em Hermite normal form}.
    However, both decision problems are easily proved to be
    computationally equivalent, see \cite[Lemma 3]{L11a}. The proof
    applies directly to \FFDP{}.
  \end{remark}

A {\em distinguisher} between two distributions $\code{D}_{0}$ and
$\code{D}_{1}$ is a probabilistic polynomial time (PPT) algorithm
$\code{A}$ that takes as input an oracle $\code{O}_{b}$ corresponding
to a distribution $\code{D}_{b}$ with $b\in \{0, 1\}$ and outputs an
element $\code{A}(\code{O}_{b})\in\{0, 1\}$.
\iftoggle{llncs}{\ \newline}{Consider the following approach for solving a decision problem
between two distributions $\code{D}_{0}$ and $\code{D}_{1}$, pick
$b\sample \{0, 1\}$ and answer $b$ regardless of the input. This
algorithm solves this problem with probability $1/2$ which is not
interesting. The efficiency of an algorithm $\code{A}$ solving a
decision problem is measured by the difference between its probability
of success and $1/2$. The relevant quantity to consider is the
{\em advantage} defined as:

\[
  \adv_{\code{A}}(\code{D}_{0}, \code{D}_{1}) \eqdef
  \dfrac{1}{2} \left( \Prob(\code{A}(\code{O}_{b}) =
    1 \mid b = 1) - \Prob(\code{A}(\code{O}_{b}) = 1 \mid b = 0) \right)
\]
where the probabilities are computed over the internal randomness of
$\code{A}$, a uniform $b \in \{0,1\}$ and inputs according to a distribution
$\code{D}_{b}$. The advantage of a distinguisher
$\code{A}$ measures how good it is to solve a distinguishing problem.
Indeed, it is classical fact that:
\[
  \Prob(\code{A}(\code{O}_{b}) = b) = \frac{1}{2} + \adv_{\code{A}}(\code{D}_{0}, \code{D}_{1}).
\]
\begin{remark}
  Even if it means answering $1-\code{A}(\code{O}_{b})$ instead of
  $\code{A}(\code{O}_{b})$, the advantage can always be assumed to be
  a positive quantity.
\end{remark}

\noindent {\bf A module version.} Instead of considering one secret
$\bfs\in\OO_{K}/\gothP$, we could use multiple secrets
$(\bfs_{1},\dots,\bfs_{d})\in \left( \OO_{K}/\gothP\right )^{d}$. This
generalization has been considered in lattice-based cryptography under
the terminology Module-\LWE{} \cite{LS15}, where the secret can be
thought as an element of $\OO_{K}^{d}$ which is a free
$\OO_{K}$-module of rank $d$, before a reduction modulo $\gothP$ on
each component. This would yield the following definition.

\begin{definition}
[\MFFDP{} Distribution]\label{def:MFFDP_distribution}
Let $d\ge 1$ be an integer. A sample
$(\bfa, \bfb) \in \left(\OO_{K}/\gothP\right)^{d}\times \OO_{K}/\gothP$
is distributed according to the \MFFDP{} distribution modulo $\gothP$
with secret
$\bfs \eqdef (\bfs_{1},\dots, \bfs_{d})\in \left(\OO_{K}/\gothP\right)^{d}$
and error distribution $\distrib$ over $\OO_{K}/\gothP$ if

  \begin{itemize}
\item $\bfa$ is uniformly distributed over
          $\left(\OO_{K}/\gothP\right)^{d}$,
  \item $\bfb = \sum_{i=1}^{d}\bfa_{i}\bfs_{i}+\bfe\in \OO_{K}/\gothP$ where $\vec{e}$ is distributed according to $\distrib$.
\end{itemize}
\end{definition}

The search and decision problems associated to \MFFDP{} can be defined
as a natural generalization of Problems \ref{pb:FFDP_search} and
\ref{pb:FFDP_decision}.

\iftoggle{llncs}{\begin{definition}}{\begin{problem}}
[\MFFDP{}, Search version]
\label{pb:MFFDP_search} Let $\bfs\in(\OO_{K}/\gothP)^{d}$ be a
collection of elements of $\OO_{K}/\gothP$ called the secrets, and let
$\distrib$ be a probability distribution over $\OO_{K}/\gothP$. An
instance of the \MFFDP{} problem consists in an oracle giving access
to independent samples $(\bfa, \bfb)$ from the \MFFDP{} distribution
with secrets $\bfs$ and error distribution $\distrib$. The goal is to
recover $\bfs$.  \iftoggle{llncs}{\end{definition}}{\end{problem}}

\iftoggle{llncs}{\begin{definition}}{\begin{problem}}
		[\MFFDP{}, Decision version]
        \label{pb:MFFDP_decision}
        Let $\bfs$ be drawn uniformly at random in
        $(\OO_{K}/\gothP)^{d}$ and let $\distrib$ be a probability
        distribution over $\OO_{K}/\gothP$.  Define $\code{D}_{0}$ to
        be the uniform distribution over
        $(\OO_{K}/\gothP)^{d}\times\OO_{K}/\gothP$, and $\code{D}_{1}$
        to be the \MFFDP{} distribution with secrets $\bfs$ and error
        distribution $\distrib$. Furthermore, let $b$ be a uniform
        element of $\{0,1\}$.
    Given access to an oracle $\code{O}_{b}$ providing samples from
     distribution $\code{D}_{b}$, the goal of the decision \MFFDP{}
    is to recover $b$.
    \iftoggle{llncs}{\end{definition}}{\end{problem}}
}

{\bf \noindent Search to decision reduction.} \iftoggle{llncs}{We are now ready to present our main theorem.}{There is an obvious
reduction from the decision to the search version of \FFDP{}. Indeed,
if there exists an algorithm $\code{A}$ that given access to the
$\ffd{\bfs, \distrib}$ distribution is able to recover the secret
$\bfs$, then it yields to a distinguisher between
$\ffd{\bfs, \distrib}$ and the uniform distribution. The converse
reduction needs more work. However, due to the strong analogy between
function and number fields, our proof is in fact essentially the same
as in \cite{LPR10,L11a}. More precisely, we have the following
theorem.}

\begin{theorem}[Search to decision reduction for \FFDP{}]\label{thm:main}
  Let ${K}/{\Fq(T)}$ be a Galois function field of degree $n$ with
  field of constants $\Fq$, and denote by $\OO_{K}$ its ring of
  integers.  Let $Q(T)\in\Fq[T]$ be an irreducible
  polynomial. Consider the ideal $\mathfrak{P}\eqdef Q\OO_{K}$. Assume
  that $\mathfrak{P}$ does not ramify in $\OO_K$, and denote by $f$
  its inertia degree.  Let $\distrib$ be a probability distribution
  over ${\OO_{K}}/{\gothP}$, closed under the action of
  $\Gal({K}/{\Fq(T)})$, meaning that if $\vec{e}\sample\distrib$, then for
  any $\sigma \in \Gal(K/\Fq(T))$, we have
  $\sigma(\vec{e})\sample\distrib$. Let $\bfs\in {\OO_{K}}/{\gothP}$.

  Suppose that we have an access to
  $\ffd{\vec{s}, \distrib}$ and there exists a distinguisher between
  the uniform distribution over ${\OO_{K}}/{\goth{P}}$ and the \FFDP{}
  distribution with uniform secret and error distribution $\distrib$,
  running in time $t$ and having an advantage $\eps$. Then there
  exists an algorithm that recovers $\bfs\in {\OO_{K}}/{\gothP}$ (with
  an overwhelming probability in $n$) in time
  \[
    O\left( \frac{n^{4}}{f^{3}}\times \frac{1}{\varepsilon^{2}} \times
      q^{f \deg(Q)}\times t\right).
  \]
\end{theorem}

\begin{remark}
	We have assumed implicitly in the statement of the theorem that we
    have an efficient access to the Galois group of ${K}/{\Fq(T)}$
    and its action can be computed in polynomial time.
\end{remark}

\begin{remark}
  There are many degrees of freedom in the previous statement: choice
  of the function field $K$ (and on the degree $n$), choice of the
  polynomial $Q$ (and on $f$ and $\deg(Q)$). For our instantiations,
  we will often choose the ``modulus'' $Q$ to be a linear polynomial
  ($\deg(Q) = 1$) and $K$ will be a (subfield of) a cyclotomic
  function field.
\end{remark}

\begin{remark} Due to the continuity of error distributions used in
  lattice-based cryptography, a technical tool called the {\em
    smoothing parameter} was introduced by Micciancio and Regev in
  \cite{MR04}. It characterizes how a Gaussian distribution
  is close to uniform, both modulo the lattice,
  and is ubiquitously used in reductions. However, in the function
  field setting, we do not need to introduce such a tool because the
  error distribution is discrete and already defined on the quotient
  ${\OO_{K}}/{\gothP}$.
\end{remark}
\iftoggle{llncs}{ }{
\begin{remark} In \cite{LS15}, Langlois and Stehlé proved a search to
  decision reduction for the module version of \LWE{}. The idea is to
  use the distinguisher in order to retrieve the secrets one by one.
  Their proof applies {\em mutatis mutandis} to \MFFDP{}, resulting in
  a time overhead of $d$, where $d$ denotes the rank of the underlying
  module, {\em i.e.} the number of secrets. The main change is in the {\em
    guess and search} step (Step 3 in the proof presented in Section
  \ref{sec:StD}) where the randomization is applied on only one
  component of $\bfa$ to recover one secret, and repeating the process
  $d$ times (one for each secret). More precisely, for \MFFDP{}, the
  running time claimed in Theorem \ref{thm:main} should be replaced
  with
  \[
    O\left( d\times \frac{n^{4}}{f^{3}} \times
      \frac{1}{\varepsilon^{2}} \times q^{f \deg(Q)} \times t \right).
  \]
\end{remark}
}

\iftoggle{llncs}{
  \begin{remark}[\MFFDP{}]
    Instead of considering one secret $\bfs\in\OO_{K}/\gothP$, we
    could use multiple secrets
    $(\bfs_{1},\dots,\bfs_{d})\in \left( \OO_{K}/\gothP\right )^{d}$.
    The goal is now to recover the secrets from samples $(\bfa, \bfb)$
    with $\bfa = (\bfa_{1},\dots,\bfa_{d})$ uniformly distributed over
    $(\OO_{K}/\gothP)^{d}$ and
    $\bfb = \inner{\bfa, \bfs} + \bfe = \sum_{i=1}^{d}\bfa_{i}\bfs_{i}+\bfe$
    with $\bfe \sample \distrib$. This generalization has been
    considered in lattice-based cryptography under the terminology
    Module-\LWE{} \cite{LS15}, because the secret can be thought as an
    element of $\OO_{K}^{d}$ which is a free $\OO_{K}$-module or rank
    $d$, before a reduction modulo $\gothP$ on each component.

    Following \cite[Section 4.3]{LS15}, it is possible to adapt Theorem
    \ref{thm:main} ; the search to decision only yielding an overhead
    of $d$ (the number of secrets). The running time would now be
    \[
      O\left( d\times \frac{n^{4}}{f^{3}} \times \frac{1}{\varepsilon^{2}} \times q^{f \deg(Q)} \times t \right).
    \]
  \end{remark}
}{}

\iftoggle{llncs}{ {\noindent \em Sketch of Proof of Theorem \ref{thm:main}.}  The proof of this
  Theorem is very similar to the one for Ring--\LWE{}
  and lattices \cite{LPR10}. It uses four steps that we quickly
  describe. Let $\gothP = \gothP_{1}\dots\gothP_{r}$, where $r = n/f$,
  be the factorisation of $\gothP$ in prime ideals.
		\begin{description}[align = left,leftmargin=*]
        \item[Step 1. Worst to Average Case.\label{desc:step1}] In the
          definition of Problem~\ref{pb:FFDP_decision} the secret
          $\bfs$ is supposed to be {\em uniformly} distributed over
          ${\OO_{K}}/{\gothP}$, while in the search version the secret
          is {\em fixed}. This can easily be addressed, for any sample
          $(\bfa, \bfb) \sample \ffd{\bfs, \distrib}$ with fixed
          secret $\bfs$, it is enough to pick
          $\bfs'\sample {\OO_{K}}/{\gothP}$ and output
          $(\bfa, \bfb + \bfa \bfs')$.  \newline

          \item[Step 2. Hybrid argument.\label{desc:step2}] sample
            $(\bfa, \bfb)$ is said to be distributed according to the
            hybrid distribution $\code{H}_{i}$ if it is of the form
            $(\bfa', \bfb' + \bfh)$ where
            $(\bfa', \bfb') \sample \ffd{\bfs, \distrib}$ and
            $\bfh\in {\OO_{K}}/{\gothP}$ is uniformly distributed
            modulo $\gothP_{j}$ for $j\le i$ and $\vec{0}$ modulo the
            other factors. Such an $\bfh$ can easily be constructed
            using the Chinese Remainder Theorem. In particular, for
            $i=0$, $\vec{h}$ is $\vec{0}$ modulo all the factors of
            $\gothP$, therefore $\vec{h}=\vec{0}$ and
            $\code{H}_{0} = \ffd{\bfs, \distrib}$. On the other hand,
            when $i=r$, the element $\vec{h}$ is uniformly distributed
            over ${\OO_{K}}/{\gothP}$, therefore $\code{H}_{r}$ is
            {\em exactly} the uniform distribution over
            ${\OO_{K}}/{\gothP}$.

			By a hybrid argument, we can turn a distinguisher
            $\code{A}$ for \FFDP{} with advantage $\varepsilon$, into
            a distinguisher between
            $(\code{H}_{i_{0}},\code{H}_{i_{0}+1})$ for some $i_{0}$
            with advantage $\geq \varepsilon/r$. Everything is
            analysed as if we knew this index $i_{0}$. In practice we
            can run $\code{A}$ concurrently with all the $r$
            instances.  \newline

          \item[Step 3. Guess and search. \label{desc:step3}] The idea
            is to perform an exhaustive search in
            ${\OO_{K}}/{\gothP_{i_{0}+1}}$ and to use $\code{A}$ to
            recover
            $\widehat{\vec{s}} \eqdef \bfs\mod\gothP_{i_{0}+1}$. Let
            $\vec{g}_{i_{0}+1} \mathop{=}\limits^{?}
            \widehat{\vec{s}}$ be our guess and set
            $\vec{g} \equiv \vec{g}_{i_{0}+1} \mod \gothP_{i_{0}+1}$
            and $\vec{0}$ otherwise. For each sample
            $(\vec{a},\vec{b})$ we compute
            $\vec{a}' \eqdef \vec{a} + \vec{v}$ and
            $\vec{b}' \eqdef \vec{b} + \vec{h} + \vec{v}\vec{g} =
            \bfa' \bfs + \bfe + \bfh'$ where
            $\bfh' = \bfh + \bfv (\bfg - \bfs)$ with
            $\vec{v}\equiv \vec{v}_{i_{0}+1}$ uniform modulo
            $\gothP_{i_{0}+1}$, and $\vec{h}$ uniform modulo the
            $\gothP_{j}$ for $j \leq i_{0}+1$ and $\vec{0}$
            otherwise. One can verify that,
            \[\left\lbrace
			\begin{array}{ll}
              \bfh' \equiv \vec{h}_{j}& \mod \gothP_{j} \text{ for } j\le i_{0}
              \\
              \bfh' \equiv (\vec{g}_{i_{0}+1} -
              \widehat{\vec{s}}) \vec{v}_{i_{0}+1}\ & \mod \gothP_{i_{0}+1} \\
              \bfh' \equiv \vec{0} & \mod \gothP_{j} \text{ for } j > i_{0}+1.
              \\
			\end{array}\right.
            \]
            If the guess $\vec{g}_{i_{0}+1}$ is correct, $(\bfa', \bfb')$ is distributed
            according to $\code{H}_{i_{0}}$. Otherwise, it is
            distributed according to $\code{H}_{i_{0}+1}$ because
            $\bfv_{i_{0}+1}$ is uniformly distributed over
            $\OO_{K}/\gothP_{i_{0}+1}$ which is a field. The
            distinguisher will succeed with probability
            $1/2+\eps/r>1/2$. It suffices to repeat the procedure
            $\Theta((r/\eps)^{2})$ times, and do a majority voting to
            know whether the guess $\vec{g}_{i_{0}+1}$ is correct or
            not. We do that for all the $q^{f\deg(Q)}$ possible
            guesses.  \newline
          \item[Step 4. Galois action. \label{desc:step4}] Since
            $K/\Fq(T)$ is Galois, for any $j\neq i_{0}$ we take
            $\sigma\in\Gal(K/\Fq(T))$ such that
            $\sigma(\gothP_{j}) = \gothP_{i_{0}}$. Now,
            $(\sigma(\bfa), \sigma(\bfa)\sigma(\bfs)+\sigma(\bfe))
            \sample \ffd{\sigma(\bfs), \distrib}$ because $\distrib$
            is Galois invariant. The above procedure enables to
            recover $\sigma(\bfs)\mod \gothP_{i_{0}}$. Applying
            $\sigma^{-1}$ yields $\bfs\mod \gothP_{j}$. Therefore, we
            are able to recover $\bfs\mod\gothP_{j}$ for any $j$. To
            compute the full secret $\bfs$ it remains to use the
            CRT.\qed
		\end{description}
		 }{}

 \iftoggle{llncs}{}
{\section{Search to Decision Reductions: Proof of Theorem
  \ref{thm:main}}\label{sec:StD}

\iftoggle{llncs}{ {\bf \noindent Basic definition: advantage.}
	Consider the following approach for solving a distinguishing problem
	between two distributions $\code{D}_{0}$ and $\code{D}_{1}$, pick
	$b\sample \{0, 1\}$ and answer $b$ regardless of the input. This
	algorithm solves this problem with probability $1/2$ which is not
	interesting. The efficiency of an algorithm $\code{A}$ solving a
	decision problem is measured by the difference between its probability
	of success and $1/2$. The relevant quantity to consider is the
	{\em advantage} defined as:

	\[
	\adv_{\code{A}}(\code{D}_{0}, \code{D}_{1}) \eqdef
	\dfrac{1}{2} \left( \Prob(\code{A}(\code{O}_{b}) =
	1 \mid b = 1) - \Prob(\code{A}(\code{O}_{b}) = 1 \mid b = 0) \right)
	\]
	where the probabilities are computed over the internal randomness of
	$\code{A}$, a uniform $b \in \{0,1\}$ and inputs according to a distribution
	$\code{D}_{b}$. The advantage of a distinguisher
	$\code{A}$ measures how good it is to solve a distinguishing problem.
	Indeed, it is classical fact that:
	\[
	\Prob(\code{A}(\code{O}_{b}) = b) = \frac{1}{2} + \adv_{\code{A}}(\code{D}_{0}, \code{D}_{1}).
	\]
\begin{remark}
		Even if it means answering $1-\code{A}(\code{O}_{b})$ instead of
		$\code{A}(\code{O}_{b})$, the advantage can always be assumed to be
		a positive quantity.
	\end{remark}

}{}

\iftoggle{llncs}{The proof of Theorem \ref{thm:main} is very similar to the one for Ring-\LWE{} and
	lattices. It uses four steps that we describe.}{In this section, we give a proof of Theorem \ref{thm:main}. It is very similar to the one for Ring-\LWE{} and
lattices. It uses four steps that we describe.} Combining them provides
the aforementioned result.
The main line of proof is as follows. We
use an hybrid argument to reduce the search domain, and then proceed
to an exhaustive search using the distinguisher to recover $\bfs$
modulo all the factors of $\gothP$. Finally, using the Chinese
Remainder Theorem (CRT) one can recover $\bfs$ completely. The key
point here is the action of the Galois group on the primes and that
the error distribution is Galois invariant.

Let $\gothP = \gothP_{1}\dots\gothP_{r}$ be the decomposition of
$\gothP$, we have
$$
r = n/f
$$
where we used the assumptions over $\gothP$ and ${K}/{\Fq(T)}$ made in
Theorem~\ref{thm:main}, namely that ${K}/{\Fq(T)}$ is a Galois
extension of degree $n$ and unramified at $\gothP$ with inertia degree
$f$. \newline

{\bf \noindent Step 1: Worst to Average Case.} Recall that in the
definition of Problem~\ref{pb:FFDP_decision} the secret $\bfs$ is
supposed to be {\em uniformly} distributed over ${\OO_{K}}/{\gothP}$,
while in the search version the secret is {\em fixed}. In other words,
the decision problem is somehow an {\em average} case problem, while
the search version should work in {\em any} case. Fortunately, this
can easily be addressed by randomizing the secret. Indeed, for any
sample $(\bfa, \bfb) \sample \ffd{\bfs, \distrib}$ with fixed secret
$\bfs$, if $\bfs'\sample {\OO_{K}}/{\gothP}$, then
$(\bfa, \bfb + \bfa \bfs')$ is now a sample from
$\ffd{\bfs+\bfs', \distrib}$ with secret $\bfs+\bfs'$ uniformly
distributed over ${\OO_{K}}/{\gothP}$.  \newline

{\bf \noindent Step 2: Hybrid argument.} Let $\code{A}$ be the
distinguisher between the uniform distribution over
${\OO_{K}}/{\goth{P}}$ and the \FFDP{} distribution with uniform
secret and error distribution $\distrib$, running in time $t$ and
having an advantage $\eps$. We use a simple hybrid argument to prove
that $\code{A}$ can also distinguish in time $t$ between two
consecutive {\em hybrid} distributions with advantage at least
$\eps/r$.

The factorization of $\gothP$ is $\gothP_{1}\dots\gothP_{r}$. A
sample $(\bfa, \bfb)$ is said to be distributed according to the
hybrid distribution $\code{H}_{i}$ if it is of the form
$(\bfa', \bfb' + \bfh)$ where
$(\bfa', \bfb') \sample \ffd{\bfs, \distrib}$ and
$\bfh\in {\OO_{K}}/{\gothP}$ is uniformly distributed modulo $\gothP_{j}$
for $j\le i$ and $\vec{0}$ modulo the other factors. Such an $\bfh$
can easily be constructed using the Chinese Remainder Theorem. In
particular, for $i=0$, $\vec{h}$ is $\vec{0}$ modulo all the factors
of $\gothP$, therefore $\vec{h}=\vec{0}$ and
$\code{H}_{0} = \ffd{\bfs, \distrib}$. On the other hand, when $i=r$,
the element $\vec{h}$ is uniformly distributed over ${\OO_{K}}/{\gothP}$,
therefore $\code{H}_{r}$ is {\em exactly} the uniform distribution
over ${\OO_{K}}/{\gothP}$.

\begin{lemma}[Hybrid argument]\label{lem:hybrid}
  There exists $i_{0}$ such that
  $\adv_{\code{A}}(\code{H}_{i_{0}}, \code{H}_{i_{0}+1}) \ge \frac{\eps}{r}$.
\end{lemma}

\begin{proof}
  By definition, $\adv_{\code{A}}(\code{H}_{0}, \code{H}_{r}) = \eps$.
  Furthermore, the following equality holds:
  \[
    \adv_{\code{A}}(\code{H}_{0}, \code{H}_{r}) =
    \sum_{i=0}^{r-1}\adv_{\code{A}}(\code{H}_{i}, \code{H}_{i+1}).
  \]
  Therefore, it exists $i_{0}\in \llbracket 0, r-1\rrbracket$ such
  that
  $\adv_{\code{A}}(\code{H}_{i_{0}},\code{H}_{i_{0}+1}) \ge \frac{\adv_{\code{A}}(\code{H}_{0}, \code{H}_{r})}{r} = \frac{\eps}{r}$.
\end{proof}

This hybrid argument has shown the existence of an $i_{0}$ such that
$\code{A}$ has an advantage $\varepsilon/r$ for distinguishing
distributions $\code{H}_{i_{0}}$ and $\code{H}_{i_{0}+1}$. In what
follows, everything is analysed as if we knew this index $i_{0}$. In
practice we can run $\code{A}$ concurrently with all the $r$ instances
$(\code{H}_{i},\code{H}_{i+1})$'s. Computations on the right index
$i_{0}$ will output the secret $\vec{s}$ (which can be verified) as it
will be explained afterward. Therefore, our reduction will output
$\vec{s}$ with a ``resource overhead'' given by at most a factor $r$.
\newline

{\bf \noindent Step 3: Guess and search.}  Given $i_{0}$ such as in
Lemma \ref{lem:hybrid}. The idea is to perform an exhaustive search in
${\OO_{K}}/{\gothP_{i_{0}+1}}$ and to use $\code{A}$ to recover
$\bfs\mod\gothP_{i_{0}+1}$.

\begin{lemma}\label{lem:guess_and_search}
  Let $\code{A}$ be a distinguisher with advantage $\delta$ between
  hybrid distributions $\code{H}_{i_{0}}$ and $\code{H}_{i_{0}+1}$, with
  secret $\bfs$, running in time $t$. Then there exists an algorithm
  $\code{B}$ that recovers $\bfs\mod \gothP_{i_{0}+1}$ with overwhelming probability in $n$ in time
  $O\left(q^{f\deg(Q)} \times \frac{n}{\delta^{2}}\times t \right)$.
\end{lemma}

\begin{proof}
  Our algorithm will proceed with a {\em guess and search} technique
  using the distinguisher $\code{A}$ in hand. The idea is to {\em
    guess} the value of $\bfs\mod\gothP_{i_{0}+1}$ and transform any
  sample $(\bfa, \bfb)\sample\ffd{\bfs, \distrib}$ into a sample of
  $\code{H}_{i_{0}}$ if the guess is correct, and into a sample of
  $\code{H}_{i_{0}+1}$ if the guess is incorrect.

  \noindent {\bf Transformation:}
  Let $\vec{g}_{i_{0}+1}\in {\OO_{K}}/{\gothP_{i_{0}+1}}$. It will be our guess for
  $\widehat{\vec{s}}\eqdef\bfs\mod \gothP_{i_{0}+1}$. Let us consider now the following operations

  \begin{itemize}
    \item Take $\bfg\in {\OO_{K}}/{\gothP}$ such that
          $\bfg\equiv \vec{g}_{i_{0}+1}\mod \gothP_{i_{0}+1}$ and
          $\bfg\equiv \vec{0}\mod \gothP_{j}$ for $j\neq i_{0}+1$.
    \item Sample $\vec{h}_{j}\sample {\OO_{K}}/{\gothP_{j}}$ for
          $1\le j\le i_{0}$ and take $\bfh\in {\OO_{K}}/{\gothP}$ such that
          $\bfh\equiv \vec{h}_{j}\mod \gothP_{j}$ for $1\le j\le i_{0}$ and
          $\bfh\equiv \vec{0}\mod \gothP_{j}$ for $j\ge i_{0}+1$.
\item Sample $\vec{v}_{i_{0}+1}\sample {\OO_{K}}/{\gothP_{i_{0}+1}}$ and take
          $\bfv\in {\OO_{K}}/{\gothP}$ such that
          $\bfv\equiv \vec{v}_{i_{0}+1}\mod \gothP_{i_{0}+1}$ and
          $\bfv\equiv \vec{0}\mod \gothP_{j}$ for $j\neq i_{0}+1$.
  \end{itemize}

  All
  those operations can be done via the CRT. Now, for each sample
  $(\bfa, \bfb\eqdef \bfa\bfs + \bfe)\sample\ffd{\bfs,\distrib}$,
  define $(\bfa', \bfb')$ with
  \[
  \bfa' \eqdef \bfa + \bfv \quad\mbox{and}\quad
  \bfb' \eqdef \bfb + \bfh +\bfv\bfg.
  \]
  Note that for each sample $(\bfa, \bfb)$, the corresponding $\bfa'$ is
  still uniformly distributed over ${\OO_{K}}/{\gothP}$ and $\bfb'~=~\bfa'\bfs + \bfe + \bfh'$ with
  $\bfh' \eqdef \bfh + (\bfg-\bfs)\bfv$. Furthermore, $\vec{h}$ verifies:
  \[\left\lbrace
      \begin{array}{ll}
        \bfh' \equiv \vec{h}_{j}& \mod \gothP_{j} \text{ for } j\le i_{0} \\
        \bfh' \equiv (\vec{g}_{i_{0}+1} - \widehat{\vec{s}}) \vec{v}_{i_{0}+1}\ & \mod \gothP_{i_{0}+1} \\
        \bfh' \equiv \vec{0} & \mod \gothP_{j} \text{ for } j > i_{0}+1. \\
      \end{array}\right.
    \]
In particular, $\bfh'$ is uniformly distributed modulo
    $\gothP_{j}$ for $j\le i_{0}$ and $\vec{0}$ modulo $\gothP_{j}$ for
    $j > i_{0}+1$.

    Now, if the guess is correct, meaning
    $\vec{g}_{i_{0}+1}=\widehat{\vec{s}}$, then
    $\bfh'\equiv 0\mod \gothP_{i_{0}+1}$, hence $(\bfa',\bfb')$ is
    distributed according to $\code{H}_{i_{0}}$. On the other hand, if
    the guess is incorrect,
    $(\vec{g}_{i_{0}+1}-\widehat{\vec{s}})\neq \vec{0}$ in
    ${\OO_{K}}/{\gothP_{i_{0}+1}}$. But $\OO_{K}$ is a Dedekind domain and
    $\gothP_{i_{0}+1}$ is a prime ideal, therefore it is also maximal
    and ${\OO_{K}}/{\gothP_{i_{0}+1}}$ is in fact a {\em field}. Since
    $\vec{v}_{i_{0}+1}$ is uniformly distributed in
    ${\OO_{K}}/{\gothP_{i_{0}+1}}$, so is
    $(\vec{g}_{i_{0}+1}-\widehat{\vec{s}}) \vec{v}_{i_{0}+1}$. In
    particular, $\bfh'$ is also uniformly distributed modulo
    $\gothP_{i_{0}+1}$. Hence, $(\bfa',\bfb')$ is distributed
    according to $\code{H}_{i_{0}+1}$.

    The algorithm $\code{B}$ proceeds as follows: for each
    $(\bfa,\bfb)\sample\ffd{\bfs, \psi}$, it applies the previous
    transformation to get a sample $(\bfa', \bfb')$, and then uses the
    distinguisher $\code{A}$. Repeating the procedure $m$ times (for each guess $\vec{g}_{i_{0}+1}$), for
     $m$ large enough, and doing a majority voting allows to
    recover $\widehat{\bfs}$ with overwhelming probability. More
    precisely, it relies on the use of the Chernoff bound.

    \begin{proposition}[Chernoff bound]
      Let $(X_{j})_{1\le j\le m}$ be $m$ independent Bernouilli random
      variables with parameter $1/2+\delta$. Let $X \eqdef \sum_{j=1}^{m}X_{j}$. Then
      \[
        \Prob\left(X \le \frac{m}{2}\right) \le e^{-2m\delta^{2}}.
      \]
    \end{proposition}

    Consider $m$ trials of the guess and search procedure, and let
    $X_{j}$ denote the indicator random variable that the $j$-th run
    returns the correct value. Since $\code{A}$ has distinguishing
    advantage $\delta$, $X_{j}$ is a Bernouilli with parameter
    $\frac{1}{2}+\delta$.

    After $m$ trials, the procedure fails if and only if more than
    $m/2$ runs are wrong. By Chernoff bound, the probability that it
    happens is less than $e^{-2m\delta^{2}}$. Therefore, by choosing
    $m\ge \ln(\frac{1}{\mu})\frac{1}{2\delta^{2}}$, the above
    procedure returns the correct guess with probability at least
    $1-\mu$. Therefore if one sets $\mu = 2^{-\Theta(n)}$ (to get an
    overwhelming probability of success), it is enough to choose $m$
    as $\Theta\left( \frac{n}{\delta^{2}}\right)$. It enables to
check if our guess
  $\widehat{\vec{s}} = \vec{s} \mod \gothP_{i_{0}+1}$ is correct or
  not with overwhelming probability. To recover
  $\vec{s} \mod \gothP_{i_{0}+1}$ it remains to try all the possible
  guesses $\widehat{\vec{s}} \in {\mathcal{O}_{K}}/{\gothP_{i_{0}+1}}$.
  But the size of ${\mathcal{O}_{K}}/{\gothP_{i_{0}+1}}$ is given by
  $q^{f\deg(Q)}$, which yields the claimed time complexity.
\end{proof}

{\bf \noindent Step 4: Action of the Galois group.}
Until Step $3$, we are able to recover the secret $\bfs$ modulo one
of the factors. In order to recover the full secret, we use
the Galois group $G\eqdef \Gal({K}/{\Fq(T)})$. This last part is
{\em crucial} for the reduction to work. Recall that $G$ acts transitively
on the set of prime ideals above $\mathfrak{p}$, \ie for every
$i\neq j$, there exists $\sigma\in G$ such that
$\sigma\left( \gothP_{i} \right) = \gothP_{j}$.

\begin{lemma}
  Fix $\bfs\in {\OO_{K}}/{\gothP}$. Let $1\le i\le r$ and let $\code{A}$
  be an algorithm running in time $t$, and recovering
  $\bfs\mod \gothP_{i}$ by making queries to an oracle for
  $\ffd{\bfs, \distrib}$.
  Then there exists an algorithm $\code{B}$ running in time
  $O\left( t\times r \right)$ that recovers the full secret $\bfs$.
\end{lemma}

\begin{proof}
  We build $\code{B}$ as follows: for every factor $\gothP_{j}$ of
  $\gothP$, it chooses $\sigma\in\Gal({K}/{\Fq(T)})$ such that
  $\sigma(\gothP_{j}) = \gothP_{i}$. Then, for each sample
  $(\bfa, \bfb)\sample \ffd{\bfs, \distrib}$, it runs $\code{A}$ on
  the input $(\sigma(\bfa), \sigma(\bfb))$ to recover an
  element $\vec{s}_{j}$ and stores
  $\sigma^{-1}(\vec{s}_{j})$.

  Note that $\Gal({K}/{\Fq(T)})$ keeps the uniform distribution over
  ${\OO_{K}}/{\gothP}$. In particular, for every sample
  $(\bfa, \bfb)\sample \ffd{\bfs,\distrib}$, the corresponding
  $\sigma(\bfa)$ is also uniformly distributed over
  ${\OO_{K}}/{\gothP}$.  Furthermore, $\bfb = \bfa \bfs + \bfe$ with
  $\bfe \sample \distrib$. Therefore,
  $\sigma(\bfb) = \sigma(\bfa) \sigma(\bfs) +
  \sigma(\bfe)$.  But $\distrib$ is Galois invariant by
  assumption, and hence $\sigma(\bfe)$ is also distributed
  according to $\distrib$.  In particular,
  $(\sigma(\bfa), \sigma(\bfb))$ is a valid sample of
  $\ffd{\sigma(\bfs), \distrib}$.

  Now, our algorithm $\code{A}$ is able to recover
  $\vec{s}_{j} \eqdef \sigma(\bfs) \mod \gothP_{i}$ in time
  $t$, and
  $$\sigma^{-1}(\vec{s}_{j}) = \sigma^{-1}\left(\sigma(\bfs) \mod \gothP_{i}\right) = \bfs
  \mod \sigma^{-1}\left(\gothP_{i}\right) = \bfs \mod
  \gothP_{j}.$$
  Therefore, we are able to recover
  $\bfs \mod \gothP_{j}$ for any $1 \leq j \leq r$. To compute the
  full secret $\vec{s}$ it remains to use the Chinese Remainder
  Theorem.
  The running time of this full procedure is given by a $O(t\times r)$
  which concludes the proof.
\end{proof}

 }
\section{Cyclotomic function fields and the Carlitz module}
\label{sec:carlitz}
In Section \ref{sec:FFDP}, we introduced the generic problem \FFDP{}
and noticed that our search to decision reduction\iftoggle{llncs}{}{ given in Section
\ref{sec:StD}} needed Galois function fields. In \cite{LPR10}, it was
proposed to use cyclotomic number fields to instantiate the Ring--\LWE{}
problem. Here, we propose to instantiate \FFDP{} with the function
field analogue, namely {\em Carlitz} extensions. We give a self
contained presentation of the theory of Carlitz extensions. The
interested reader can refer to \cite[ch. 12]{R02}, \cite{NX01} and
the excellent survey \cite{Conrad-Carlitz} for further reference.

Carlitz extensions are function fields analogues of the cyclotomic
extensions of $\QQ$. A dictionary summarizing the similarities is
given in Table \ref{fig:cyclotomic_carlitz_analogy}. These extensions
were discovered by Carlitz in the late 1930s but the analogy was not
well known until the work of his student Hayes who studied them in
\cite{H74} to give an explicit construction of the abelian extensions
of the rational function field $\Fq(T)$ and prove an analogue of the
usual Kronecker-Webber theorem which states that any abelian extension
of $\QQ$ are subfields of cyclotomic number fields. This result was
generalized in the following years with the work of Drinfeld and Goss
to yield a complete solution to Hilbert twelfth problem in the
function field setting. In the number field setting, such an explicit
construction is only known for abelian extensions of $\QQ$ (cyclotomic
extensions), imaginary quadratic number fields (via the theory of
elliptic curves with complex multiplication).

The first idea that comes to mind when one wants to build cyclotomic
function fields is to adjoin roots of unity to the field $\Fq(T)$.
However, roots of unity are already {\em algebraic} over $\Fq$. In
other words, adding them only yields so--called {\em extensions of
  constants}.

\begin{example}
  Let $\zeta_{n}$ be an $n$--th root of unity in $\Fq(T)$. Note that
  it belongs to some {\em finite} extension of $\Fq$. Let $\Fqm$ be
  the extension of $\Fq$ of minimal degree such that
  $\zeta_{n}\in \Fqm$ (it can be $\Fq$ itself). Then
  \[
    \Fq(T)[\zeta_{n}] = \Fqm(T),
  \]
  and the field of constants of $\Fq(T)[\zeta_{n}]$ is $\Fqm$.
\end{example}

However, in our reduction setting, such extensions will only increase
the size of the search space in Step~$3$. More precisely, if $K$ is an
algebraic extension of $\Fq(T)$, the constant field of $K$ is always a
subfield of $\OO_{K}/\gothP$ for any prime ideal $\gothP$ of
$\OO_{K}$. But recall that in our search to decision reduction, we
need to do an exhaustive search in this quotient $\OO_{K}/\gothP$, so
we need it to be as small as possible. Henceforth, we cannot afford
constant field extensions. For Carlitz extensions, this will be
ensured by Theorem~\ref{thm:Carlitz_constant}.

\iftoggle{llncs}{ }{
\begin{example}
  As a matter of example, consider the polynomial $T^{2}+T+1$ over
  $\F2$. It is irreducible. Let $\zeta_{3}\in\F4$ be one of its roots.
  It is a cube root of $1$. Now, consider the field extension
  $K \eqdef \F2(T)(\zeta_{3}) = \F4(T)$ and let $\OO_{K}$ be the
  integral closure of $\F2[T]$ in $K$. The prime ideal
  $\mathfrak{p}~\eqdef~(T^{2}+T+1)$ of $\F2[T]$ splits into two prime
  ideals $\gothP_{1}$ and $\gothP_{2}$ in $\OO_{K}$. But
  $\OO_{K}/\gothP_{1} =~\OO_{K}/\gothP_{2} =~\F4 =~\F2[T]/\gothp$ and
  we do not win anything by considering the extension $K/\F2(T)$.
\end{example}
}

\subsection{Roots of unity and torsion}

As mentioned in the beginning of this section, it is not sufficient to
add roots of unity. One has to go deeper into the algebraic structure
that is adjoined to $\QQ$. Indeed, the set of all $m$--th roots of
unity, denoted by $\mu_{m} \subset \CC$, turns out to be an abelian group under multiplication. Moreover,
$\mu_{m}$ is in fact {\em cyclic}, generated by any {\em primitive}
root of unity.

In commutative algebra, abelian groups are {\em $\ZZ$-modules}. Here
the action of $\ZZ$ is given by exponentiation: $n\in\ZZ$ acts on
$\zeta\in\mu_{m}$ by $n\cdot \zeta \eqdef \zeta^{n}$. This action of
$\ZZ$ can in fact be extended to all $\overline{\QQ}^{\times}$. When
working with modules over a ring, it is very natural to consider the
{\em torsion elements},
\ie elements of the module that are annihilated by an element of the
ring. The torsion elements in the $\ZZ$--module
$\overline{\QQ}^{\times}$ are the $\zeta\in\overline{\QQ}^{\times}$
such that $\zeta^{m}=1$ for some $m>0$; these are precisely the roots
of unity. In other words, the cyclotomic number fields are obtained by
adjoining to $\QQ$ torsions elements of the $\ZZ$--module
$\overline{\QQ}^{\times}$.

Under the analogy summed up in Table \ref{fig:ff_analogy}, replacing
$\ZZ$ by $\Fq[T]$ and $\QQ$ by $\Fq(T)$, we would like to consider
some $\Fq[T]$--module and adjoin to $\Fq(T)$ the torsion
elements. Note that $\Fq[T]$--modules are in particular $\Fq$--vector
spaces\iftoggle{llncs}{, hence the action of $\Fq[T]$ should be linear.}{. The
natural candidate could be $\overline{\Fq(T)}$ with $\Fq[T]$ acting by
multiplication. However, the torsion elements are not very
interesting. Indeed if there is some $f \in \overline{\Fq(T)}$ and
some $a \in \Fq[T]\setminus \{0\}$ such that $af =0$ then, $f = 0$.
Therefore, for the usual action by multiplication, only $0$ is a
torsion element. Thus, we need to define {\em another} $\Fq[T]$--module
structure, in the same way that we did not consider the natural action
of $\ZZ$ by multiplication.} This new module structure can be defined
using so called {\em Carlitz polynomials}: for each polynomial
$M\in\Fq[T]$, we define its Carlitz polynomial $[M](X)$ as a
polynomial in $X$ with coefficients in $\Fq[T]$, and $M\in\Fq[T]$ will
act on $\alpha\in\overline{\Fq(T)}$ by
$M\cdot \alpha \eqdef [M](\alpha)$\iftoggle{llncs}{ with $[M](\alpha+\beta) = [M](\alpha)+[M](\beta)$.}{. In the literature, the notation
$\alpha^{M}$ can also be found to emphasize the analogy with the
action of $\ZZ$ by exponentiation, but it can be confusing.
}
\iftoggle{llncs}{}{In the same way that the action of $\ZZ$ was
multiplicative: $(\alpha\beta)^{n} = \alpha^{n}\beta^{n}$, the action
of $\Fq[T]$ will be {\em additive}:
$[M](\alpha + \beta) = [M](\alpha) + [M](\beta)$.} In other words,
$[M](X)$ should be an {\em additive polynomial}. In positive
characteristic this can easily be achieved by considering
$q$--polynomials, \ie polynomials whose monomials are only $q$--th
powers of $X$, namely of the form
\[
  P(X) = p_{0}X + p_{1}X^{q} + \cdots + p_{r}X^{q^{r}}.
\]

\iftoggle{llncs}{ }{
\begin{remark}
  $q$--polynomials with coefficients in some finite field $\Fqm$
  are also used in coding theory to build so called {\em rank metric codes}.
  However, here we consider $q$--polynomials with coefficients in
  $\Fq[T]$.
\end{remark}
}

\subsection{Carlitz polynomials}
The definition of Carlitz polynomial will proceed by induction and
linearity. Define $[1](X) \eqdef X$ and $[T](X) \eqdef X^{q} +TX$. For
$n\ge 2$, define

\[
  [T^{n}](X) \eqdef [T]([T^{n-1}](X)) = [T^{n-1}](X)^{q} + T[T^{n-1}](X).
\]

\vspace{\baselineskip}

\noindent Then, for a polynomial $M = \sum_{i=0}^{n}a_{i}T^{i}\in\Fq[T]$, define
$[M](X)$ by forcing $\Fq$--linearity:

\[
  [M](X) \eqdef \sum_{i=0}^{n}a_{i}[T^{i}](X).
\]

\begin{example}We have,
\begin{itemize}[label=\textbullet]
  \item $[T^{2}](X) = [T](X^{q} + TX) = X^{q^{2}}+(T^{q} + T)X^{q} + T^{2}X$
  \item $[T^{2}+T+1](X) = [T^{2}](X) + [T](X) + [1](X) = X^{q^{2}}+(T^{q}+T+1)X^{q} + (T^{2}+T+1)X$
\end{itemize}
\end{example}

By construction, Carlitz polynomials are additive polynomials, and
$\Fq$--linear. Furthermore, for two polynomials $M, N\in\Fq[T]$,
$[MN](X) = [M]([N](X)) = [N]([M](X))$. In particular, Carlitz
polynomials commute with each other under composition law, which is
not the case in general for $q$--polynomials.

\subsection{Carlitz module}
Endowed with this $\Fq[T]$--module structure, $\overline{\Fq(T)}$ is called the
{\em Carlitz module}.

\begin{definition}For $M\in\Fq[T]$, $M\neq 0$, let
  $\Lambda_{M} \eqdef \{\lambda \in \overline{\Fq(T)} \mid [M](\lambda) = 0\}$.
  This is the module of $M$--torsion of the Carlitz module.
\end{definition}

\begin{example}\label{ex:LambdaT}
  $\Lambda_{T} = \{\lambda \in \overline{\Fq(T)} \mid \lambda^{q} + T\lambda = 0\} = \{0\} \cup \{\lambda \mid \lambda^{q-1} = -T\}$.
\end{example}

In the same way that $\mu_{m}$ is an abelian group (\ie a
$\ZZ$--module), note that $\Lambda_{M}$ is also a submodule of the
Carlitz module: for $\lambda\in\Lambda_{M}$ and $A\in \Fq[T]$,
$[A](\lambda) \in \Lambda_{M}$. In particular, $\Lambda_{M}$ is an
$\Fq$--vector space.

\begin{example}
  The module $\Lambda_{T}$ defined in Example~\ref{ex:LambdaT} is an
  $\Fq$--vector space of dimension $1$. In particular, for
  $\lambda\in\Lambda_{T}$, and $A\in\Fq[T]$, $[A](\lambda)$ must be a
  multiple of $\lambda$. In fact the Carlitz action of $A$ on
  $\lambda$ is through the constant term of $A$: writing
  $A = TB + A(0)$ we have
  \[
    [A](\lambda) = [TB + A(0)](\lambda) = [B](\underbrace{[T](\lambda)}_{=0}) + A(0)[1](\lambda) = A(0)\lambda.
  \]
\end{example}

More generally, even if in general $\Lambda_{M}$ is not of dimension
$1$ over $\Fq$, it is always a {\em cyclic} $\Fq[T]$--module: as an
$\Fq[T]$--module it can be generated by only one element. This is
specified in the following theorem.

\begin{theorem}[{\cite[Lemma 3.2.2]{NX01}}] There exists
  $\lambda_{0}\in\Lambda_{M}$ such that
  $\Lambda_{M} = \{ [A](\lambda_{0}) \mid A\in\Fq[T]/(M) \}$ and the
  generators of $\Lambda_{M}$ are the $[A](\lambda_{0})$ for all $A$
  prime to $M$. The choice of a generator yields a non canonical
  isomorphism $\Lambda_{M}\simeq \Fq[T]/(M)$ as $\Fq[T]$--modules.
\end{theorem}

\begin{remark}
  The previous theorem needs to be related to the cyclotomic case:
  given the choice of a primitive $m$--th root of unity, there is a
  group isomorphism between $\mu_m$ and $\ZZ/m\ZZ$. Moreover all the
  $m$--th roots of unity are of the form $\zeta^{k}$ for
  $k\in\llbracket 0,m-1\rrbracket$ and the generators of $\mu_{m}$ are the
  $\zeta^{k}$ for $k$ prime to $m$.
\end{remark}

\subsection{Carlitz extensions} Recall that the cyclotomic number
fields are obtained as extensions
  of $\QQ$ generated by the elements of $\mu_m$. In the similar
fashion, for a polynomial $M\in\Fq[T]$, let
\[
  K_{M} \eqdef \Fq(T)(\Lambda_{M}) = \Fq(T)(\lambda_{M}),
\]
where
$\lambda_{M}$ is a generator of $\Lambda_{M}$. One of the most
important fact about the cyclotomic number field $\QQ(\zeta_{m})$ is
that it is a finite Galois extension of $\QQ$, with Galois group
isomorphic to $(\ZZ/m\ZZ)^{\times}$. There is an analogue statement
for the Carlitz extensions.

\begin{theorem}[{\cite[Th.~3.2.6]{NX01}}]\label{thm:Carlitz_Galois}
  Let $M\in\Fq[T]$, $M\neq 0$. Then $K_{M}$ is a finite Galois
  extension of $\Fq(T)$, with Galois group isomorphic to
  $(\Fq[T]/(M))^{\times}$. The isomorphism is given by
  \[
    \map{(\Fq[T]/(M))^{\times}}{\Gal(K_{M}/\Fq(T))}{A}{\sigma_{A},}
\]
  where $\sigma_{A}$ is completely determined by
  $\sigma_{A}(\lambda_{M}) = [A](\lambda_{M})$.
\end{theorem}

\begin{remark}
In particular, Carlitz extensions are {\em abelian}.
\end{remark}

Another important fact about cyclotomic extensions is the simple
description of their ring of integers. Namely, for $K = \QQ(\zeta_{m})$,
we have $\OO_{K} = \ZZ[\zeta_{m}] = \ZZ[X]/(\Phi_{m}(X))$ where
$\Phi_{m}$ denotes the $m$--th cyclotomic polynomial. This property
also holds for Carlitz extensions.

\begin{theorem}[{\cite[Th.~2.9]{R02}}]
  Let $\OO_{M}$ be the integral closure of $\Fq[T]$ in $K_{M}$. Then
  $\OO_{M} =~\Fq[T][\lambda_{M}]$. In particular, let
  $P(T, X) \in \Fq[T][X]$ be the minimal polynomial of $\lambda_{M}$.
  Then,
  \[
    K_{M} = \fract{\Fq(T)[X]}{(P(T, X))} \quad \text{and} \quad
    \OO_{M} = \fract{\Fq[T][X]}{(P(T, X))}.
  \]
\end{theorem}

\begin{example}\label{ex:LambdaTbis}
  Reconsider Example~\ref{ex:LambdaT} and the module
  $\Lambda_{T} = \{0\} \cup \{ \lambda \mid \lambda^{q-1} = -T\}$. The
  polynomial $X^{q-1}+T$ is Eisenstein in $(T)$ and therefore is
  irreducible. Hence,
  \[
    K_{T} = \fract{\Fq(T)[X]}{(X^{q-1}+T)}.
  \]
  Moreover it is Galois, with Galois group
  $(\Fq[T]/(T))^{\times} \simeq \Fq^{\times}$. A non-zero element
  $a\in\Fq^{\times}$ will act on $f(T,X)\in K_{T}$ by
  \[
    a\cdot f(T, X) \eqdef f(T, [a](X)) = f(T, aX).
  \]
  The integral closure of $\Fq[T]$ in $K_{T}$ is
  \[
    \OO_{T} \eqdef \fract{\Fq[T][X]}{(X^{q-1}+T)}
  \]
  and
  \begin{equation}\label{eq:O_T/(T+1)}
    \fract{\OO_{T}}{((T+1)\OO_{T})} =
    \fract{\Fq[T][X]}{(T+1, X^{q-1}+T)} =
    \fract{\Fq[X]}{(X^{q-1}-1)}.
  \end{equation}
\end{example}

Finally, the following theorem characterizes the splitting behaviour
of primes in Carlitz extensions. A very similar result holds for
cyclotomic extensions.

\begin{theorem}[{\cite[Th.~12.10]{R02}}]\label{thm:carlitz_splitting}
  Let $M\in\Fq[T]$, $M\neq 0$, and let $Q\in\Fq[T]$ be a monic,
  irreducible polynomial. Consider the Carlitz extension $K_{M}$ and
  let $\OO_{M}$ denote its ring of integers. Then,
  \begin{itemize}[label=\textbullet]
    \item If $Q$ divides $M$, then $Q\OO_{M}$ is totally ramified.
    \item Otherwise, let $f$ be the smallest integer $f$ such that
          $Q^{f} \equiv 1 \mod M$. Then $Q\OO_{M}$ is unramified and
          has inertia degree $f$. In particular, $Q$ splits completely
          if and only if $Q\equiv 1 \mod M$.
  \end{itemize}
\end{theorem}

Note that in Ring--\LWE{}, the prime modulus $q$ is often chosen such
that $q\equiv 1 \mod m$ so that it splits completely in the cyclotomic
extension $\QQ(\zeta_{m})$.

\begin{example}
  In the previous example, $T+1 \equiv 1 \mod T$ and therefore $(T+1)$
  splits completely in $\OO_{T}$. Indeed,
  \[
    \fract{\OO_{T}}{((T+1)\OO_{T})} = \fract{\Fq[X]}{(X^{q-1}-1)} = \prod_{\alpha \in \Fq^{\times}} \fract{\Fq[X]}{(X-\alpha)}
  \]
  is a product of $q-1$ copies of $\Fq$.
\end{example}

It is crucial for the applications that the constant field of $K$
be not too big because, in the search--to--decision reduction, it
determines the search space in Step $3$ of the proof of
Theorem~\ref{thm:main}. The following non-trivial theorem gives the
field of constants of Carlitz extensions.

\begin{theorem}[{\cite[Cor.~of ~Th.~12.14]{R02}}]\label{thm:Carlitz_constant} Let $M\in\Fq[T]$,
  $M\neq 0$. Then $\Fq$ is the full {\em constant field} of $K_{M}$.
\end{theorem}

The similarities between Carlitz function fields and cyclotomic number
fields are summarized in Table \ref{fig:cyclotomic_carlitz_analogy}.

\begin{table}[!ht]
  \centering
      \[
      \begin{array}{|c|c|}
        \hline
        \QQ & \Fq(T) \\
        \ZZ & \Fq[T] \\
        \text{Prime numbers } q \in \ZZ & \text{Irreducible polynomials } Q\in\Fq[T]\\
            & \\
        \mu_{m} = \ideal{\zeta} \simeq \ZZ/m\ZZ \text{ (groups) }& \Lambda_{M} =\ideal{\lambda}\simeq \Fq[T]/(M) \text{ (modules) }\\
            & \\
        d \mid m \Leftrightarrow \mu_{d} \subset \mu_{m} \text{ (subgroups) } & D\mid M \Leftrightarrow \Lambda_{D} \subset \Lambda_{M} \text{ (submodules) }\\
            & \\
        a \equiv b \mod m \Rightarrow \zeta^{a} = \zeta^{b} & A \equiv B \mod M \Rightarrow [A](\lambda) = [B](\lambda)\\
& \\
        K = \QQ[\zeta] & K = \Fq(T)[\lambda] \\
        \OO_{K} = \ZZ[\zeta] & \OO_{K} = \Fq[T][\lambda] \\
            & \\
        \Gal(K/\QQ) \simeq (\ZZ/m\ZZ)^{\times} & \Gal(K/\Fq(T))\simeq (\Fq[T]/(M))^{\times}\\
            & \\
        \textbf{Cyclotomic} & \textbf{Carlitz}\\
        \hline
      \end{array}
    \]
    \caption{\label{fig:cyclotomic_carlitz_analogy} Analogies between
      cyclotomic and Carlitz}
\end{table}

 \section{Applications}\label{sec:applications}
In the current section, we present two applications of our proof
techniques. It provides search to decision reductions to generic
problems whose hardness assumption has been used to assess the
security of some cryptographic designs. The first application concerns
Oblivious Linear Evaluation (OLE) which is a crucial primitive for
secure multi-party computation. The second one is an authentication
protocol called \lapin{}. Both designs rely on the hardness of
variants of the so-called Learning Parity with Noise (\LPN{}) problem.

\subsection{\LPN{} and its structured variants}
Let us start this subsection by the definitions of the distribution
that is involved in the \LPN{} problem.

\begin{definition}[Learning Parity with Noise (\LPN{}) distribution]
  Let $k$ be a positive integer, $\bfs\in \Fq^k$ be a uniformly
  distributed vector and $p \in [0, \frac{1}{2})$. A sample
  $(\vec{a},b)\in\Fq^{k} \times \Fq$ is distributed according to the
  \LPN{} distribution with secret $\vec{s}$ if
  \begin{itemize}[label=\textbullet]
    \item $\vec{a}$ is uniformly distributed over $\Fq^{k}$;
    \item $b \eqdef \langle \vec{a},\vec{s} \rangle + e$ where
          $\inner{ \cdot\ \! , \cdot }$ denotes the canonical inner
          product over $\Fq^k$ and $e$ is a $q$--ary Bernouilli random
          variable with parameter $p$, namely $\mathbb{P}(e=0) = 1-p$
          and $\mathbb{P}(e=a) = \frac{p}{q-1}$ for
          $a\in\Fq^{\times}$.
  \end{itemize}
	A sample drawn according to this distribution will be denoted
  $ (\bfa, \inner{ \bfa, \bfs } + e) \sample \DLPN{\bfs}{p}$.
\end{definition}

\begin{remark}\label{rem:q-ary Bernouilli}
  This definition is a generalization of the usual \LPN{} distribution
  defined over $\F2$. In this situation, the error distribution is a
  usual Bernouilli: $\Prob(e=0) = 1 - p$ and $\Prob(e=1) = p$.
\end{remark}

\iftoggle{llncs}{}{
\begin{remark}\label{rem:ntuple_LPN}
  Sometimes in the literature, the distribution is directly defined
  for $n$ samples, leading to $(\bfG, \vec{s}\cdot \bfG + \bfe)$ where
  $\bfG$ is drawn uniformly at random over the space
  $\Fq^{k \times n}$ of $k \times n$ matrices whose coefficients lie
  in $\Fq$ and $\bfe \eqdef (e_1, \dots, e_n)$ where the $e_i$'s are
  independent Bernouilli random variables with parameter $p$.
\end{remark}

The security of many cryptosystems in the literature rests on the LPN
assumption which informally asserts that it is hard to distinguish a
sample $(\bfa, \inner{ \bfa, \bfs } + e) \sample \DLPN{\bfs}{p}$ from
a sample $(\bfa, t)$ where both $\bfa$ and $t$ are drawn uniformly at random.

\begin{remark}\label{rem:LPN=decoding}
  Note that, according to Remark~\ref{rem:ntuple_LPN}, when
  considering a fixed number $n$ of samples, the LPN assumption is
  nothing but the decision version of the decoding problem, namely
  distinguishing noisy codewords of a random code from uniformly
  random vectors.
\end{remark}
}

Similarly to the \LWE{} problem, structured versions of \LPN{} have been
defined (\cite{HKLPP12,DP12,BCGIKS20}).

\begin{definition}[Ring--\LPN{} distribution]
  Fix a positive integer $r$, a public polynomial $f(X) \in~\Fq [X]$
  of degree $r$ and $\vec{s} \in \Fq[X]/(f(X))$ be a uniformly distributed
  polynomial. A sample $(\vec{a}, \vec{b})$ is distributed according to the
  \RLPN{} distribution with secret $\vec{s}$ if
\begin{itemize}[label=\textbullet]
	\item $\vec{a}$ is drawn
	uniformly at random over $\Fq[X]/(f(X))$;
  \item $\vec{b} \eqdef \vec{a}\vec{s}+\vec{e}$ where
    $\vec{e} \eqdef e_0+ e_1X + \cdots + e_{r - 1}X^{r - 1} \in
    \Fq[X]/(f(X))$ has coefficients $e_{i}$'s which are independent
    $q$--ary Bernouilli random variables with parameter $p$.
\end{itemize}
A sample drawn according to this distribution will be denoted
$ (\vec{a},\vec{a}\vec{s}+\vec{e}) \sample \DRLPN{\vec{s}}{p}$.
\end{definition}

Note that the map
\[
  \map{\fract{\Fq[X]}{(f(X))}}{\fract{\Fq[X]}{(f(X))}}{\vec{m}(X)}{\vec{a}(X)\vec{m}(X)
    \mod f(X)}
\]
can be represented in the canonical basis by an $r \times r$ matrix
$\bfA$. Using this point of view, one sample of \RLPN{} can be
regarded as $r$ specific samples of \LPN{}.

\iftoggle{llncs}{\ }{
\begin{definition}[Module--\LPN{} distribution]
  Fix positive integers $r$ and $d$, a public polynomial
  $f(X) \in \Fq [X]$ of degree $r$ and $\bfs$ be uniformly distributed
  over $(\Fq[X]/(f(X)))^{d}$. A sample $(\bfa, \bfb)$ is distributed
  according to the \MLPN{} distribution with secrets $\bfs$ if
\begin{itemize}
	\item $\bfa$ is drawn
	uniformly at random over $(\Fq[X]/(f(X)))^{d}$;
	\item $\vec{b} \eqdef \inner{\bfa, \bfs} + \vec{e} = \sum_{i=1}^{d}\vec{a}_{i}\vec{s}_{i}+\vec{e}$,
        where
        $\vec{e} \eqdef e_0+ e_1X + \cdots + e_{r - 1}X^{r - 1} \in \Fq[X]/(f(X))$
        has coefficients $e_{i}$'s which are independent Bernouilli
        random variables with parameter $p$.
\end{itemize}
A sample drawn according to this distribution will be denoted
$ (\bfa,\vec{b}) \sample \DMLPN{\vec{s}}{p}$
\end{definition}

In the above definitions, the noise
distribution is chosen independently on each coefficient of
$\vec{e}(X)\in\Fq[X]/(f(X))$. One can also consider the situation where the
coefficients of $\vec{e}(X)$ are chosen to form a vector of $\Fq^{r}$ of
fixed Hamming weight $t$. This point of view is closer to the usual
decoding problem, and is the one adopted in \cite{BCGIKS20}.

To conclude, note that the choice of the noise in the Ring-- and
Module--\LPN{} distribution is made relatively to the canonical basis
${(X^{i})}_{0\leq i \leq r-1}$ since it seems to be the most natural
one. However, one could have made another choice. This will be
discussed in the sequel.
}

\iftoggle{llncs}{}{
\subsection{Relation with decoding problems}
According to Remark~\ref{rem:LPN=decoding}, distinguishing $n$ samples
of an \LPN{} distribution $\DLPN{\vec{s}}{p}$ from the uniform
distribution is equivalent to distinguishing noisy codewords of a
known random code from uniformly random words. It can be regarded as
the decision version of the decoding problem for a code of fixed
dimension (here the length of the secret $\vec{s}$) and whose length
$n$ goes to infinity and hence whose rate goes to zero. Similarly,
distinguishing samples from a distribution $\DRLPN{\bfs}{p}$ and a
uniform one is equivalent to the decision version of the decoding
problem of structured codes whose basis is block-wise defined as
\[
  \begin{pmatrix}
    \bfA_1 & \cdots & \bfA_m
  \end{pmatrix}
\]
where, for any $i \in \llbracket 1,n\rrbracket$, $\bfA_i$ is the
matrix representation in the canonical basis of $\Fq[X]/(f(X))$ of the
multiplication by some random element $\bfa_i \in \Fq[X]/(f(X))$. In
particular, considering the case $f(X) = X^{\ell} - 1$, we recover,
the decoding problem for $\ell$--quasi--cyclic codes.
\newline
}

{\bf \noindent Search to decision.}
Here we present search to decision reductions in two different
settings corresponding to two choices of the modulus $f(X)$ in the
Ring--\LPN{} problem. Both have been used in the literature for
specific applications that are quickly recalled.
\newline

{\bf \noindent A $q$--ary version of \textup{Ring--LPN} with a totally
  split modulus $f$.}
In \cite{BCGIKS20}, the authors introduce
Ring--\LPN{} over
the finite field $\Fq$ and with a modulus $f$ which is totally split,
{\ie} has distinct roots, all living in the ground field $\Fq$.
\newline

{\bf \noindent Motivation: Oblivious Linear Evaluations for secure
  Multi Party Computation (MPC).}
A crucial objective in modern secure MPC is to be able to generate
efficiently many random pairs $(u, r), (v, s)$ such that $u, v, r$ are
uniformly distributed over $\Fq$, with the correlation $uv = r+s$.

In \cite{BCGIKS20}, the authors propose a construction of such pairs
$(\bfu, \bfr), (\bfv, \bfs)$ of elements in a ring $\mathcal{R}$,
where $\mathcal R = \Fq[X]/(f(X))$ such that $f$ is split with simple
roots in $\Fq$. Using the Chinese remainder Theorem, one deduces
$\deg f$ pairs $(u_i, r_i), (v_i, s_i)$ with
$u_i, v_i, r_i, s_i \in \Fq$. The pseudo-randomness of
$\vec{u},\vec{v}$ rests on the hardness of the Ring--\LPN{}
assumption.

\bigskip
{\bf \noindent Search to decision reduction in the \cite{BCGIKS20}-case.}
Consider the case of Ring--\LPN{} over $\mathcal R = \Fq[X]/(f(X))$,
where
\[
  f(X) \eqdef \prod_{a \in \Fq^{\times}} (X-a) = X^{q-1}-1.
\]
Let us re-introduce the Carlitz function field of
Examples~\ref{ex:LambdaT} and~\ref{ex:LambdaTbis}, namely
\[
  K_T = \fract{\Fq(T)[X]}{(X^{q-1} + T)}.
\]
According to Equation \eqref{eq:O_T/(T+1)} in
Example~\ref{ex:LambdaTbis}, we have
\[
  \fract{\OO_T}{(T+1)\OO_T} \simeq \fract{\Fq[X]}{(X^{q-1}-1)},
\]
which is precisely the ring we consider for the Ring \LPN{} version of
\cite{BCGIKS20}.  Therefore, instantiating our \FFDP{} problem with
this function field, modulus $T+1$, ideal $\gothP \eqdef (T+1)\OO_{K}$ and
applying Theorem~\ref{thm:main}, we directly obtain the following search to
decision reduction.

\begin{theorem}[Search to decision reduction for totally-split Ring--LPN{}]\label{thm:StD_LPN}
  Let $K_{T}$ be the Carlitz extension of $T$--torsion over $\Fq$, and
  denote by $\OO_{T}$ its ring of integers. Consider the ideal
  $\gothP\eqdef~(T+1)\OO_{K_{T}}$. Then $\gothP$ splits completely in
  $q-1$ factors $\gothP_{1}\dots\gothP_{q-1}$ and
  \[
    \OO_{K}/\gothP \simeq \prod_{i=1}^{q-1}\OO_{K}/\gothP_{i} \simeq \Fq \times \cdots \times \Fq.
  \]
  Let $\distrib$ denote the uniform distribution over polynomials in
  $\Fq[X]/(X^{q-1}-1)$ of fixed Hamming weight, or the $q$--ary
  Bernouilli distribution. Let $\bfs\in\Fq[X]/(X^{q-1}-1)$. Suppose
  that we have access to $\ffd{\bfs,\distrib}$ and that there exists a
  distinguisher between the uniform distribution over
  $\Fq[X]/(X^{q-1}-1)$ and $\ffd{\bfs, \distrib}$ with uniform secret
  and error distribution $\distrib$, running in time $t$ and having
  advantage $\eps$.

  Then there exists an algorithm that recovers $\bfs$ with overwhelming probability (in $q$) in time
  $$O\left( q^{5}\times \frac{1}{\varepsilon^{2}} \times t\right).$$
\end{theorem}

\begin{proof}
The only thing that remains to be proved is that the error
distribution is Galois invariant.
According to Theorem~\ref{thm:Carlitz_Galois} and
Example~\ref{ex:LambdaTbis}, the Galois group of $K_T/\Fq(T)$ is
isomorphic to $(\Fq[T]/(T))^{\times} \simeq \Fq^{\times}$. Furthermore,
we proved that an element $b\in\Fq^{\times}$ acts on
$f(T, X)\in K_{T}$ by
\[
  b\cdot f(T, X) = f(T, [b](X)) = f(T, bX).
\]
\iftoggle{llncs}{}{
For this example it can actually be understood directly using {\em
  Kummer theory} (see \cite[Proposition 3.7.3]{S09}). Indeed, $\Fq$,
and hence $\Fq(T)$, contains the $(q-1)$--th roots of unity. Moreover,
$K_{T}$ is nothing but the extension of $\Fq(T)$ spanned by a
primitive $(q-1)$--th root of $-T$. Therefore, it is a {\em Kummer
  extension} and the action of the Galois group is characterized by
$X\mapsto \zeta\cdot X$ for every $(q-1)$--th root of unity $\zeta$.
But here, the set of $(q-1)$--th roots of unity is {\em precisely}
$\Fq^{\times}$.
}The Galois action on $K_T$ and $\OO_T$ induces an action of
$\Fq^{\times}$ on
\[
  \fract{\OO_T}{(T+1)\OO_T} \simeq \fract{\Fq[X]}{(X^{q-1}-1)}
\]
by
$b \cdot m(X) \eqdef m(b X)$. Note that, this operation has no
incidence on the Hamming weight of $m$: it actually {\em does not
  change its Hamming support}. Therefore, we easily see here that
Galois action keeps the noise distribution invariant.
\end{proof}

\begin{remark} Note that our search to decision reduction could have
  been performed here without introducing the function field and only
  considering the ring $\Fq[X]/(X^{q-1} - 1)$. Recall that the first
  ingredient of the reduction is to decompose this ring by the Chinese
  Remainder Theorem. Here it would give the product
  $\prod_{a \in \Fq^{\times}} \Fq[X]/(X-a)$. The final step of the
  reduction requires the introduction of a group action which induces
  a permutation of the factors in
  $\prod_{a \in \Fq^{\times}} \Fq[X]/(X-a)$. It is precisely what the
  group action $b \cdot m(X) = m(b X)$ does: it sends the
  factor $\Fq[X]/(X-a)$ onto $\Fq[X]/(X - b^{-1}a)$.  However,
  introducing this action on the level of $\Fq[X]/(X^{q-1}-1)$ does
  not look very natural. It turns out that the introduction of
  function fields permits to interpret this action in terms of a
  Galois one.
\end{remark}

\iftoggle{llncs}{
  \begin{remark}
    If we replace the Carlitz extension $K$ by some subfield of
    invariants under the action of a given subgroup of the Galois
    group, it is possible to extend the result to the case where
    $f(X) = \prod_{a \in H}(X-a)$ where $H$ is some subgroup of
    $\Fq^{\times}$. It is even possible to treat the case where the
    roots of $f$ form a coset of a given subgroup of $\Fq^{\times}$.
  \end{remark}
}
{
\noindent {\bf Case H is a strict subgroup of $\Fq^{\times}$.} Now
  assume the polynomial $f$ has the form
\[
  f(X) = \prod_{a\in H}(X-a)
\]
with $H$ being a strict subgroup of $\Fq^{\times}$. To instantiate our search to decision reduction we need to find a group
action that keeps the noise distribution invariant.

\begin{lemma}
  There exists a Galois function field $K$ with its ring of integers
  $\OO_{K}$ such that $H=~\Gal(K/\Fq(T))$ and
  $\OO_{K}/(T+1)\OO_{K} = \Fq[X]/(f(X))$. Moreover, the action of $H$
  keeps the Hamming support invariant.
\end{lemma}

\begin{proof}
  Consider again the Carlitz extension $K_{T}$. It has a {\em cyclic}
  Galois group $G\simeq \Fq^{\times}$ of cardinality $q-1$. Let
  $h \eqdef \#H$. It divides $q-1$. Since $G$ is cyclic, it has a {\em
    unique} subgroup $N$ of cardinality $\frac{q-1}{h}$. Let
  $L \eqdef K_{T}^{N}$ be the fixed field of $N$. Since $(T+1)$ splits
  completely in $\OO_{T}$, it also splits completely in any
  intermediate field. In particular, it splits completely in
  $\OO_{L}\eqdef \OO_{T}^{H}$ the ring of integers of $L$, we have
  $$
  \OO_{L}/(T+1)\OO_{L} \simeq \Fq[X]/(f(X)).
  $$
  Now, since $G$ is {\em
    abelian}, $N$ is {\em normal} in $G$. In particular, $L/\Fq(T)$ is
  a Galois extension, with Galois group $G/N \simeq H$. By the same
  argument as in the previous paragraph, the action of $H$ on
  $\Fq[X]/(f(X))$ only permutes the factors, but keeps the {\em
    supports} invariant. In particular, the noise distribution is not
  moved under the action of $H$.
\end{proof}

This lemma immediately implies a search to decision reduction analogue
to theorem \eqref{thm:StD_LPN}.

\begin{remark}
  When the roots of $f(X)$ do not form a subgroup of $G$, but a {\em
    coset} $bH$ instead, \ie $f(X) = \prod_{\alpha\in bH}(X-\alpha)$,
  then it suffices to perform a translation by $b$ prior:
  $X\mapsto bX$ also keeps all the distributions invariant (including
  the uniform), and $\Fq[X]/(f(X))$ is mapped onto $\Fq[X]/(g(X))$
  where $g(X) \eqdef \prod_{\alpha\in H}(X - \alpha)$, which yields
  the result by seeing the action of $H$ as arising from a Galois
  action.
\end{remark}
}

{\bf \noindent Ring--\LPN{} with a modulus $f$ splitting in irreducible
  polynomials of the same degree.}
Another cryptographic design whose security rests on the Ring--\LPN{}
assumption is an authentication protocol named \lapin{}
\cite{HKLPP12}.  In the conclusion of their article, the authors
mention that
\begin{center}
  \begin{minipage}{0.9\textwidth}
  ``{\em it would be particularly interesting to find out
  whether there exists an equivalence between the decision and the
  search versions of the problem similar to the reductions that exist
  for \LPN{} and Ring--\LWE{}}''.
  \end{minipage}
\end{center}
For this
protocol, the problem is instantiated with the binary field $\F{2}$
and with a modulus polynomial $f$ which splits as a product of $m$
distinct irreducible polynomials
\[
  f(X) = f_1(X) \cdots f_m(X).
\]
In this setting and using our techniques, we can provide a search to
decision reduction when the $f_i$'s have all the same degree $d$.
Furthermore, for the reduction to run in polynomial time, we need to
have $d = O(\log (\deg f))$. In this setting, the Chinese Reminder
Theorem entails that
\[
  \fract{\F{2}[X]}{(f(X))} \simeq \prod_{i=1}^m
  \fract{\F{2}[X]}{(f_i(X))},
\]
and the right--hand side is a product of $m$ copies of $\F{2^d}$.
Such a product can be realised as follows. Consider a function field
$K$ which is a Galois extension of $\F{2}(T)$ with Galois group $G$
and denote by $\OO_K$ the integral closure of $\F{2}[T]$ in
$K$. Suppose that the ideal $(T)$ of $\F{2}[T]$ is unramified in $\OO_K$
with inertia degree $d$. Then $T\OO_K$ splits into a product of prime
ideals:
\[
T\OO_K = \mathfrak{P}_1 \cdots \mathfrak{P}_m
\qquad
\text{and}
\qquad
  \fract{\OO_K}{T\OO_K} \simeq \prod_{i=1}^m \fract{\OO_K}{\mathfrak{P}_i},
\]
where, here again, the right--hand side is a product of $m$ copies of
$\F{2^d}$.

Next, the idea is now to apply Theorem~\ref{thm:main} in this setting.
However, there is here a difficulty since for our search to decision
reduction to hold, the noise should arise from a Galois invariant
distribution. Thus, if we want the noise distribution to be Galois
invariant we need to have a Galois invariant $\F{2}$--basis of the
algebra $\OO_K/T\OO_K$. The first question should be whether such a
basis exists. The existence of such a basis can be deduced from deep
results of number theory due to Noether \cite{Noether32,C94} and
asserting the existence of local normal integral bases at non ramified
places. Here we give a pedestrian proof resting only on basic facts of
number theory. Since this result also holds for larger finite fields,
from now on, the underlying field is not supposed to be $\F2$ anymore.

\begin{proposition}\label{prop:base_normale}
  Let $K/\Fq(T)$ be a finite Galois extension of Galois group $G$
  and $\OO_K$ be the integral closure of $\Fq[T]$ in $K$. Let
  $Q \in \F{q}[T]$ be an irreducible polynomial such that the
  corresponding prime ideal is unramified and has inertia degree $d$.
  Denote by $\mathfrak{P_1} \cdots \mathfrak{P}_m$ the decomposition
  of the ideal $Q\OO_K$. Then, $G$ acts on the finite dimensional
  algebra $\OO_K/Q\OO_K$ and there exists $\vec{x} \in \OO_K/Q\OO_K$
  such that $(\sigma (\vec{x}))_{\sigma \in G}$ is an $\Fq$--basis of
  $\OO_K/Q\OO_K$.
\end{proposition}

\begin{proof}
  Consider the decomposition group $D_{\mathfrak{P}_1/Q}$.  As
  explained Section~\ref{sec:prereqFF} and in particular in Equation
  (\ref{eq:decomp_isom}), since $Q\OO_K$ is unramified, this
  decomposition group is isomorphic to
$\Gal (\OO_K/Q\OO_K,\Fq) = \Gal(\F{q^d},\Fq)$.  This entails in
  particular that $\# D_P = d$.

  According to the Chinese Remainder Theorem,
  \[
    \fract{\OO_K}{Q\OO_K} \simeq \fract{\OO_K}{\mathfrak{P}_1} \times \cdots
    \times \fract{\OO_K}{\mathfrak{P}_m}.
  \]
  Next, from the
  Normal basis Theorem (see for instance \cite[Thm.~2.35]{LN97}), there
  exists $\vec{a} \in {\OO_K}/{\mathfrak{P}_1}$ such that
  $(\sigma(\vec{a}))_{\sigma \in D_{\mathfrak{P}_1/Q}}$ is an $\Fq$--basis of
  ${\OO_K}/{\mathfrak{P}_1}$. Now, let
  \[
    \vec{b} \eqdef (\vec{a},0,\dots,0) \in \prod_{i=1}^m \OO_K/\mathfrak{P_i}
    \simeq \OO_K / Q\OO_K.
  \]
  We claim that $(\sigma(\vec{b}))_{\sigma \in G}$ is an $\Fq$--basis of
  $\OO_K/Q\OO_K$.
  Indeed, denote by $V$ the $\Fq$--span of
  $\{\sigma(\vec{b}) ~|~ \sigma \in G\}$ and suppose that $V$ is a proper
  subspace of ${\OO_K}/{Q \OO_K}$. Then, there exists
  $i \in \llbracket 1,m \rrbracket$ such that
  \[
    V \cap  \fract{\OO_K}{\mathfrak{P}_i}
    \varsubsetneq \fract{\OO_K}{\mathfrak{P}_i},
  \]
  where we denote by ${\OO_K}/{\mathfrak{P}_i}$ the subspace
  $ \{0\} \times \cdots \times \{0\} \times
  {\OO_K}/{\mathfrak{P}_i } \times \{0\} \times \cdots \times
  \{0\}$ of $\prod_i \OO_K/\mathfrak{P}_i$.

  Since $G$ acts transitively on the $\mathfrak{P}_i$'s, there exists
  $\sigma_0 \in G$ such that
  $\sigma_0(\mathfrak{P}_1) = \mathfrak{P}_i$.  Then,
  $\sigma_0(\vec{b}) \in V \cap \OO_K/\mathfrak{P}_i$ and so does
  $\sigma \sigma_0(\vec{b})$ for any $\sigma \in D_{\mathfrak{P}_i/P}$. Since
  $V \cap {\OO_K}/{\mathfrak{P}_i} \varsubsetneq
  {\OO_K}/{\mathfrak{P}_i}$, then $\dim_{\Fq}V < d$ while
  $\# D_{\mathfrak{P}_i/P} = d$. Hence, there exist nonzero elements
  ${(\lambda_\sigma)}_{\sigma \in D_{\mathfrak{P}_i/P}} \in \Fq^d$ such that
  \begin{equation}\label{eq:Galois_relation}
    \sum_{\sigma \in D_{\mathfrak{P}_i/P}} \lambda_{\sigma} \sigma \sigma_0(\vec{b}) = 0.
  \end{equation}
  Applying $\sigma_0^{-1}$ to (\ref{eq:Galois_relation}), we get
  \[
    \sum_{\sigma \in D_{\mathfrak{P}_i/P}} \lambda_{\sigma}
    \sigma_0^{-1} \sigma \sigma_0(\vec{b}) =  0.
  \]
  As mentioned in Section~\ref{sec:prereqFF}, we
  have $\sigma_0^{-1}D_{\mathfrak{P}_i/Q}\sigma_0 = D_{\mathfrak{P}_1/Q}$
  and we deduce that the above sum is
  in ${\OO_K}/{\mathfrak{P}_1}$ and, since $\vec{a}$ is a generator of a normal
  basis of $\Fq$, we deduce that the $\lambda_{\sigma}$'s are all zero.  A
  contradiction.
\end{proof}

The previous proposition asserts the existence of a {\em normal}
$\Fq$--basis of the space ${\OO_K}/{Q\OO_K}$, \ie a Galois invariant
basis.  For any such basis, ${(\vec{b}_{\sigma})}_{\sigma \in G}$ one can
define a Galois noise distribution by sampling linear combinations of
elements of this basis whose coefficients are independent Bernouilli
random variables. Our Ring--\LPN{} distribution is hence defined as
pairs $(\bfa,\bfb) \in {\OO_K}/{Q\OO_K} \times {\OO_K }/{Q\OO_K}$ such
that $\bfa$ is drawn uniformly at random and $\bfb = \bfa\bfs+\bfe$
where $\bfe$ is a noise term drawn from the previously described
distribution.

\begin{definition}[Galois modulus]
  Let $r$ and $d$ be positive integers. A polynomial $f(X)\in\Fq[X]$
  of degree $r$ is called a Galois modulus of inertia $d$ if there
  exists a Galois function field $K/\Fq(T)$ and a polynomial
  $Q(T)\in\Fq[T]$ of degree one such that $\Fq[X]/(f(X)) \simeq \OO_{K}/Q\OO_{K}$
  and the ideal $Q\OO_{K}$ has inertia degree $d$ and does not ramify.
\end{definition}

  This definition entails that for a polynomial $f(X)\in\Fq[X]$ to be
  a Galois modulus, it needs to factorize in $\Fq[X]$ as a product of
  distinct irreducible polynomials of same degree $d$.

  Carlitz extensions permit to easily exhibit many Galois moduli of
given inertia $d$.  Indeed, let $M(T) \in \Fq[T]$ be any divisor of
$T^d-1$ which vanishes at least at one primitive $d$--th root of
unity.  Set
\[
  r \eqdef
  \frac{\#{\left(\fract{\Fq[X]}{(M(X))}\right)}^{\times}}{d}\cdot
\]
Then, any polynomial $f(X) \in \Fq[X]$ which is a product of $r$
distinct irreducible polynomials of degree $d$ is a Galois modulus.
Indeed, $\Fq[X]/(f(X))$ is isomorphic to a product of $r$ copies of
$\F{2^d}$ and, since the multiplicative order of $T$ modulo $M(T)$ is
$d$, from Theorem~\ref{thm:carlitz_splitting} so does $\OO_M/T\OO_M$.

\begin{example}\label{ex:Galois_modulus}
  The polynomial $f(X) \eqdef X^{63}+X^{7}+1 \in \F2[X]$ is a {\em
    Galois} modulus of inertia $9$. Indeed, let
  $M(T) \eqdef T^{6}+T^{3}+1$ and consider $K_{M}$ the Carlitz
  extension of $M$--torsion. Denote by $\OO_{M}$ the integral closure
  of $\F2[T]$ in $\OO_{M}$. Then $T^{9} \equiv 1 \mod M$ and $9$ is
  the smallest integer that has this property. By Theorem
  \ref{thm:carlitz_splitting}, the ideal $T\OO_{M}$ splits into $7$
  ideals $\gothP_{1},\dots,\gothP_{7}$ and has inertia $9$. Hence,
  $\OO_{M}/(T\OO_{M})\simeq \Fq[X]/(f(X)).$
\end{example}

\begin{remark} The polynomial $f(X)$ of Example
  \ref{ex:Galois_modulus} is also {\em lightness-preserving} in the
  sense of \cite[Def 2.22]{DP12} which can be used to instantiate
  Ring-LPN.
\end{remark}

We are now ready to define a new noise distribution which is Galois
invariant for Ring--\LPN{}.  We propose to consider it in \lapin{} as
it enables to apply our search to decision reduction. In the following
definition, $\code{B}$ denotes a normal basis whose existence is
ensured by Proposition~\ref{prop:base_normale}.  Note that $\code B$
need not be exactly the normal basis constructed in the proof of
Proposition~\ref{prop:base_normale}. This is discussed further, after
the statement of Theorem~\ref{thm:stdpenr}.

\begin{definition}[Normal Ring--\LPN{} distribution] Let $r, d$ be
  positive integers, $p\in [0, \frac{1}{2})$ and let
  $f(X)\in\Fq[X]$ be a {\em Galois} modulus of degree $r$ with inertia
  $d$. Denote by
$\code{B}\eqdef (\sigma(\bfc)(X))_{\sigma\in
    G_{f}}$ the normal basis of $\Fq[X]/(f(X))$ where
  $G_{f}$ is the Galois group of the related function field.

  A sample $(\bfa, \bfb)$ is distributed according to the {\em Normal}
  \RLPN{} distribution relatively to basis $\code{B}$, with secret
  $\bfs$ if
  \begin{itemize}[label=\textbullet]
    \item $\bfa$ is drawn uniformly at random over $\Fq[X]/(f(X))$;
    \item $\bfb \eqdef \bfa\bfs + \bfe$, where
          $\bfe(X) \eqdef \sum_{\sigma\in G_{f}}e_{\sigma}\sigma(\bfc)(X) \in \Fq[X]/(f(X))$
          has coefficients $e_{i}$'s which are independent $q$--ary
          Bernouilli random variables with parameter $p$.
  \end{itemize}

\end{definition}

\begin{theorem}\label{thm:stdpenr}
	The decision Ring--\LPN{} is equivalent to its search version for the normal Ring--\LPN{} distribution.
\end{theorem}

  Let us discuss further the choice of the noise distribution and
  hence that of a Galois-invariant basis.  In \cite{HKLPP12}, the
  authors discuss the case of Ring--\LPN{} when the modulus $f$ splits
  and mention that in this situation, the Ring--\LPN{} problem reduces
  to a smaller one by projecting the samples onto a factor
  ${\Fq[X]}/{(f_i(X))}$ of the algebra
  ${\Fq[X]}/{(f(X))}$. The projection onto such a factor,
  reduces the size of the inputs but increases the rate of the noise.

  It should be emphasized that the Galois invariant basis constructed
  in the proof of Proposition~\ref{prop:base_normale} yields a noise
  which is partially cancelled when applying the projection
  ${\OO_K}/{Q\OO_K} \rightarrow {\OO_K}/{\mathfrak{P}_i}$,
  hence, this choice of normal basis might be inaccurate. On
  the other hand, Proposition~\ref{prop:base_normale} is only an
  existence result and it turns out actually that a random element of
  ${\OO_K}/{Q\OO_K}$ generates a normal basis with a high
  probability.  Indeed, the existence of such a normal basis can be
  reformulated as ${\OO_K}/{Q\OO_K}$ is a free $\Fq[G]$--module
  of rank $1$ and a generator $\bfa \in {\OO_K}/{Q\OO_K}$ is an
  $\Fq[G]$--basis of ${\OO_K}/{Q\OO_K}$. Now, any other element
  of $\Fq[G]^{\times} \bfa$ is also a generator of a normal basis.
  Consequently, the probability that a uniformly random element of
  ${\OO_K}/{Q\OO_K}$ is a generator of a normal basis is
  \[
    \frac{\# \Fq[G]^\times}{\# \Fq[G]}\cdot
  \]
  If for instance, $G$ is cyclic of order $N$ prime to $q$. Then
  $X^N-1$ splits into a product of distinct irreducible factors
  $u_1\cdots u_r$ and $\Fq[G] \simeq {\Fq[X]}/{(X^N-1)} \simeq \prod_i
  {\Fq[X]}/{(u_i(X))}$. In this context, the probability that a uniformly
  random element of ${\OO_K}/{Q\OO_K}$ generates a normal basis is
  \[
    \frac{\prod_{i=1}^r (q^{\deg u_i}-1)}{q^N}\cdot
  \]

 \section*{Conclusion}
We introduced a new formalism to study generic problems useful in
cryptography based on structured codes. This formalism rests on the
introduction of function fields as counterparts of the number fields
appearing in cryptography based on structured lattices. Thanks to this
new point of view, we succeeded in producing the first search to
decision reduction in the spirit of Lyubashevsky, Peikert and Regev's
one for Ring-\LWE{}. We emphasize that such reductions were completely
absent in cryptography based on structured codes and we expect them to
be a first step towards further search to decision reductions.

If one puts into perspective our current assessment with lattice-based
cryptography, \cite{LPR10} focuses on cyclotomic number fields, and
defined the error distribution to be a Gaussian over $\RR^{n}$ through
the Minkowski embedding. Furthermore, the modulus $q$ was chosen to
split completely. Then, following this result, \cite{LS15} uses a
``switching modulus'' technique in order to relax the arithmetic
assumption on the prime modulus, so that it can be arbitrarily chosen.
Finally, the search to decision reduction has been proved in
\cite{RSW18} to hold even when the extension is not Galois, using the
Oracle with Hidden Center Problem (OHCP) technique from \cite{PRS17}.
Note that this powerful technique has been used recently to provide a
search to decision reduction in the context of NTRU \cite{PS21}. Even
though our work does not reflect these recent progresses, we believe,
as it was shown by our instantiations, that the introduction of the
function field framework paves the way for using these techniques in
the code setting in order to get a full reduction applying
to cryptosystems such as \HQC{} or \BIKE.

\bibliographystyle{alphaurl}
\newcommand{\etalchar}[1]{$^{#1}$}

\end{document}